\pdfoutput=1
\RequirePackage[T1]{fontenc}
\RequirePackage{fix-cm}
\RequirePackage{mlmodern}

\documentclass[11pt,a4paper,english]{amsart}
\usepackage{babel}
\usepackage[utf8]{inputenc}
\usepackage{amsmath}
\usepackage{amsthm}
\usepackage{amsfonts}
\usepackage{amssymb}
\usepackage{indentfirst}
\usepackage{enumitem}
\usepackage[numbers,square,sort]{natbib}
\usepackage{url}
\usepackage[hidelinks]{hyperref}
\usepackage{fullpage}

\newcommand{\ket}[1]{|#1 \rangle}

\numberwithin{equation}{section}
\theoremstyle{plain}%
\newtheorem{theorem}{Theorem}[section]
\newtheorem{proposition}[theorem]{Proposition}
\newtheorem{lemma}[theorem]{Lemma}

\theoremstyle{remark}%
\newtheorem{remark}[theorem]{Remark}
\newtheorem{example}[theorem]{Example}

\theoremstyle{definition}%

\newtheorem{definition}[theorem]{Definition}

\makeatletter
\let\addresses@\relax
\makeatother

\title[Simultaneous Reduction of Observables and Measured Qudits]{Simultaneous Reduction of Observables and Measured Qudits in
Entanglement-Assisted Quantum Local Recovery}

\author[C. Galindo, F. Hernando, and R. Matsumoto]{%
  \href{https://orcid.org/0000-0002-3908-4462}{Carlos Galindo}$^{\text{a}}$,
  \href{https://orcid.org/0000-0002-9758-2152}{Fernando Hernando}$^{\text{a}}$,
  and \href{https://orcid.org/0000-0002-5085-8879}{Ryutaroh Matsumoto}$^{\text{b}}$%
}

\date{24 November 2026}


\thanks{A preliminary version of this work, by the third author alone, was posted as arXiv:2609.23316v1.
The present manuscript extends it.}
\thanks{{\it Acknowledgments:} The first two authors were partially funded by MICIU/AEI/\-10.13039/\-501100011033 and by ``ERDF,
UE'' (grant PID2022-138906NB-C22). The third author was partially supported by ``Japan Society for the Promotion of Science''
(grant 23K10980) and by access to OpenAI models through the ChatGPT for Academic Researchers program.}

\newcommand{\affilblock}{%
\begin{center}
  {\small
    \textsuperscript{a}Instituto Universitario de Matem\'aticas y Aplicaciones de Castell\'on and Departamento de Matem\'aticas,
    Universitat Jaume I, Campus de Riu Sec, 12071 Castell\'{o}, Spain\\[0.3em]
    \textsuperscript{b}Department of Information and Communications Engineering, Institute of Science Tokyo, Japan\\[0.5em]
    Emails: \texttt{galindo@uji.es} (Galindo), \texttt{carrillf@uji.es} (Hernando), \texttt{ryutaroh@ict.e.titech.ac.jp}
    (Matsumoto)
  }
\end{center}
\vspace{1em}}
\makeatletter
\let\orig@setabstract\@setabstract
\def\@setabstract{\affilblock\orig@setabstract}
\makeatother

\begin{document}

\begin{abstract}
  Let $\mathbb{F}_q$ be the finite field with $q$ elements, $q$ being a prime power.
  For an entanglement-assisted or unassisted stabilizer
  code defined by an $\mathbb{F}_q$-linear space $C\subseteq\mathbb{F}_q^{2n}$,
  a set of erasures and an associated repair group,
  we propose a procedure to compute
  a reduced set of observables and a reduced repair group of codeword qudits
for correcting the given erasures.
  The computed set of observables has the smallest possible
  size for correcting the given erasures, within the stabilizer measurement model.
  The repair group is reduced by eliminating qudits that need not be measured.
  In classical local recovery, such symbols can be found by Gaussian elimination,
  although the resulting repair group need not be a smallest one.
  In the quantum case we eliminate an admissible set of qudits.
  Finding a largest admissible set
  is a linear matroid parity problem,
  and for Calderbank--Shor--Steane (CSS) codes defined by two linear codes
  it is a linear matroid intersection problem.
  Both problems are solvable in deterministic polynomial time.
  We give closed-form upper bounds on the size of reduced repair groups
  for general codes and for codes from Euclidean and Hermitian orthogonality,
  and, for two-code CSS codes, a sharper bound with a min--max formula.
  Based on these bounds,
  we propose quantum counterparts of the concept of dimension locality
  of classical local recovery.
\end{abstract}

\maketitle

\section{Introduction}\label{sec1}
As recent experiments have demonstrated \citep{bluvstein24,google25},
quantum error correction is an essential ingredient of large-scale quantum computers.
One of its current research directions is quantum erasure correction \citep{wu2022,kang2023}.
In both classical and quantum error correction,
an \textit{erasure} is an error whose position in a codeword is known
\citep{bennett97,grassl97,pless98}.
Since a code of minimum distance $d$ corrects up to $d-1$ erasures,
but only up to $\lfloor (d-1)/2\rfloor$ errors at unknown positions,
knowledge of the error positions considerably facilitates decoding.
Moreover, some physical platforms make it possible to identify the erased qubits of a codeword
without destroying the encoded quantum information
and without performing stabilizer measurements \citep{wu2022,kang2023}.
For these reasons, quantum erasure correction deserves to be studied in its own right,
alongside general quantum error correction.

A second research direction is the reduction of the measurements performed during decoding
\citep{perlin2023,heussen2024,veroni2024,zhou25}.
On several physical platforms, measurements disturb the measured qubits.
Consequently, reducing the number of measurements can improve the reliability
of quantum error correction and of fault-tolerant quantum computation,
although the actual effect depends on the measurement circuits and on the noise model.

It is therefore natural to study the reduction of measurements in quantum erasure correction,
a problem that involves two objectives.
The first one is to reduce the number of codeword qudits that are measured,
which is the aim of quantum local recovery
\citep{golowich25,aqlrc26,luo25,li2025optimalquantumlrcshermitian,li2025improvedboundsoptimalconstructions,qlrc24,cao2025optimalquantumrdeltalocallyrepairable,zhou2025optimalquantumrdeltalocallyrepairable,galindo26qhbch,galindo2025optimalquantumlocallyrecoverable,hqlrc26}.
The second one is to reduce the number of observables that are measured.
This second objective was recently studied in \citep{qlrsurface25eprint}.
There the measured codeword qudits are the supports of the chosen observables
\citep[Proposition 11 and Remark 12]{qlrsurface25eprint},
but their number was bounded explicitly only for generalized surface codes
\citep[Section 5]{qlrsurface25eprint}, and no reduction of a given repair group was considered.
Building on \citep{qlrsurface25eprint}, \citet{ilocality26}
addressed both objectives simultaneously and proposed a linear-algebra-based procedure
that computes a reduced set of observables and a reduced set of codeword qudits
to be measured for erasure correction.
However, the procedure of \citet{ilocality26} applies to quantum stabilizer codes
\citep{ashikhmin00,calderbank98,gottesman96,ketkar06}
defined by Hermitian dual-containing linear codes over $\mathbb{F}_{q^2}$.
For Calderbank--Shor--Steane (CSS) codes defined by two
linear codes $C_X, C_Z \subseteq \mathbb{F}_q^n$, \citet[Section 5]{ilocality26} shows by an example
that a natural candidate for a quantum information locality fails, and no procedure was given
except in the special case $C_X=C_Z$.
A procedure that applies to quantum stabilizer codes defined by arbitrary $\mathbb{F}_q$-linear symplectic spaces
is therefore desirable.

The construction of quantum stabilizer codes from linear codes over finite fields
requires dual-containing linear codes (equivalently, linear codes whose duals are self-orthogonal)
whose dimension and relative minimum distance are controlled,
and constructing such codes is in general a difficult problem.
To overcome this difficulty, \citet{brun06} introduced
entanglement-assisted quantum error-correcting codes (EAQECCs),
which dispense with the dual-containing condition at the cost of
pairs of maximally entangled quantum states shared by the encoder and the decoder.
EAQECCs were later generalized to $p$-ary codes \citep{luo17}
and to $p^m$-ary codes \citep{galindo19}, where $p$ is a prime number.
The reduction of the number of measured codeword qudits for EAQECCs,
known as entanglement-assisted quantum local recovery,
has been initiated only very recently \citep{EAQLR-gret,EAQLR-r,EAQLR-r2,eaqlrc26}.

To the best of our knowledge, a general procedure for
the simultaneous reduction of observables and measured codeword qudits
has so far been available only for entanglement-unassisted quantum stabilizer codes
constructed from Hermitian dual-containing linear codes, and for CSS codes with $C_X=C_Z$ \citep{ilocality26},
while for the remaining codes only particular cases have been considered,
such as the generalized surface codes mentioned above.
The aim of this paper is to fill this gap within the framework of entanglement-assisted quantum local recovery,
which contains the entanglement-unassisted case as the special case without maximally entangled pairs.

After reviewing the necessary background in Section \ref{sec2},
in Section \ref{sec3} we propose a procedure whose input consists of
an arbitrary $\mathbb{F}_q$-linear space $C \subseteq \mathbb{F}_q^{2n}$
defining an EAQECC of length $n$,
a set $I$ of erased positions and a repair group $K \supsetneq I$,
and whose output consists of a reduced set $J$ of measured codeword qudits
and a reduced set $O$ of observables.

The observables that the decoder may use are defined by the vectors of the symplectic dual $C^{\perp s}$ of $C$ that are supported on $K$.
Those that act as the identity on the erased qudits need not be measured,
and the remaining ones can be modified by adding vectors of the former kind.
Our procedure uses this freedom to make the observables vanish on a set $P\subseteq K\setminus I$ of qudits,
which are then not measured.

The choice of $P$ is precisely where quantum and classical local recovery differ.
For a classical linear code $E$, a repair group $K$ and a set $I$ of erasures,
the symbols of $K\setminus I$ lying outside an information set of the punctured code
$\pi_{K\setminus I}(E)$ (the restriction of $E$ to the positions in $K\setminus I$)
can be eliminated.
Since information sets are the bases of a single linear matroid,
such a set is found by Gaussian elimination.
The resulting repair group need not be a smallest one:
for a single erased position, the smallest sets of additional positions that recover it
are the sets $\operatorname{supp}(\mathbf{h})\setminus\{i\}$, where $\mathbf{h}$ has minimum Hamming weight
among the codewords of the dual code with $h_i\neq0$
(and with support contained in $K$, if the recovery must take place inside $K$)
\citep[Proposition 1 and Corollary 2]{sharp} (see also Remark \ref{rem:AI1}).
In the quantum case, a qudit can be eliminated only if
both of its symplectic coordinates are eliminated simultaneously.
We introduce \emph{admissible} sets of qudits, which can always be eliminated,
and we show that a largest admissible set is a solution of
a linear matroid parity problem
whose ground set consists of the column indices of a matrix, paired by qudit
(each qudit gives the two columns that carry its symplectic coordinates).
Linear matroid parity is a classical problem in combinatorial optimization, recalled in Subsection \ref{sec2m},
which can be solved in deterministic polynomial time \citep{lovasz81,lovasz86,gabowstallmann86,hamada26}.

Theorem \ref{thm1} summarizes the output of the procedure:
for any admissible set $P$, and given a basis of the vectors of $C^{\perp s}$ supported on $K$,
it computes with $O(|K|^3)$ operations
(here and throughout, $|\cdot|$ denotes the cardinality of a finite set)
a repair group $J$ with $I\subseteq J\subseteq K\setminus P$, and hence $|J|\le|K|-|P|$,
and the restrictions to the $n$ codeword qudits of the vectors that define
a set $O$ of $\dim\pi_I(C^{\perp s})$ observables, which is the minimum possible
within the stabilizer measurement model of \citep{qlrsurface25eprint}.
The full vectors, which also act on the $c$ qudits held by the decoder,
are then obtained with $O(|K|^2c)$ further operations from the values on that basis
of the map $\alpha$ of Subsection \ref{sec212}.
These observables correct the erasures at $I$ and,
among the $n$ codeword qudits, they act only on the qudits in $J$.
Moreover, if $P$ is a largest admissible set or the admissible set computed by Gaussian elimination
in Proposition \ref{prop:general}, this proposition gives the closed-form bound

\[
|J|-|I|\le\min\{|K\setminus I|,\dim\pi_{K\setminus I}(C)\}
\]%
for an arbitrary $C$.

In Section \ref{sec4} we bound $|J|-|I|$ in closed form for EAQECCs from Euclidean orthogonality,
namely CSS codes $C=C_X\times C_Z$ defined by two linear codes $C_X,C_Z\subseteq\mathbb{F}_q^n$,
and from Hermitian orthogonality, with or without entanglement assistance.
For CSS codes, the two coordinates of a qudit
lie in two different codes, and the matroid parity problem becomes a linear matroid intersection problem,
also solvable in deterministic polynomial time \citep{edmonds70,cunningham86,gabow96}.
Theorem \ref{thm:iotaE} shows that a largest admissible set gives

\[
|J|-|I|\le\iota_e(K\setminus I),
\]%
where, for a subset $S\subseteq\{1,\ldots,n\}$, $\iota_e(S)$ is the smallest size of a set of positions of $S$
that contains an information set of $\pi_S(C_X)$ and one of $\pi_S(C_Z)$.
It also provides a min--max formula for $\iota_e$, based on the matroid intersection theorem,
and shows that

\[
\iota_e(K\setminus I)\le\dim\pi_{K\setminus I}(C_X+C_Z),
\]%
with equality when $C_X$ and $C_Z$ are nested.
The examples of Subsection \ref{sec41} show that this inequality can be strict
and that the new bound can be optimal.
For codes defined by an $\mathbb{F}_{q^2}$-linear code $D\subseteq\mathbb{F}_{q^2}^n$
by means of Hermitian orthogonality, only one matroid is involved,
a largest admissible set is found by Gaussian elimination,
and Proposition \ref{prop:herm} gives $|J|-|I|\le\dim_{\mathbb{F}_{q^2}}\pi_{K\setminus I}(D)$.

For a classical code $E$ with a collection $\mathcal{K}$ of repair groups,
the dimension locality $\max_{K\in\mathcal{K}}\dim\pi_K(E)$ with respect to $\mathcal{K}$ \citep{sasidharan15,grezet19}
bounds the number of additional symbols needed to repair any set of erasures repaired by some $K\in\mathcal{K}$,
and it never exceeds the \emph{symbol locality} $\max_{K\in\mathcal{K}}(|K|-1)$ with respect to $\mathcal{K}$ \citep{gopalan}.

Using the above bounds, we propose quantum counterparts of this notion,
which generalize the dimension part of the definition given in \citep{ilocality26}
for codes from Hermitian orthogonality without entanglement assistance.
The \emph{dimension locality} of Definition \ref{def:ilE}, for CSS codes,
and that of Definition \ref{def:ilH}, for codes from Hermitian orthogonality, are
$\max_{K\in\mathcal{K}}\min\{\dim\pi_K(C_X+C_Z),|K|-1\}$ and
$\max_{K\in\mathcal{K}}\dim_{\mathbb{F}_{q^2}}\pi_K(D)$, respectively.
They bound $|J|-|I|$ for every $K\in\mathcal{K}$ and every set $I$ of erasures repaired by $K$,
where $J$ is computed with a largest admissible set $P$ (and $J=K$ if $I=K$),
and they never exceed the symbol locality.
Unlike in the classical case, in the CSS case the minimum with $|K|-1$ is needed (Remark \ref{rem:minK}),
whereas in the Hermitian case the inequality $\dim_{\mathbb{F}_{q^2}}\pi_K(D)\le|K|-1$ holds automatically
for every repair group $K$ for a nonempty set of erasures.
For CSS codes, the sharper bound of Theorem \ref{thm:iotaE} leads to
the refined dimension locality of Definition \ref{def:ilE2}.
Proposition \ref{prop:ilE2} shows that it also bounds $|J|-|I|$,
that it never exceeds the dimension locality of Definition \ref{def:ilE},
and that both coincide when $C_X$ and $C_Z$ are nested.
For codes from Hermitian orthogonality, admissible sets give no such refinement,
because the bound of Proposition \ref{prop:herm} cannot be improved by another admissible set $P$.

\section{Review of EAQECCs and their local recovery}\label{sec2}
In this section we collect the notation and the known results on which the paper relies,
namely, those on EAQECCs and their local recovery (Subsection \ref{sec21})
and those on the linear matroid parity and intersection problems (Subsection \ref{sec2m}).

\subsection{Quantum codes and repair groups}\label{sec21}

\subsubsection{Symplectic spaces, puncturing and shortening}\label{sec211}
Throughout the paper, $q$ is a prime power,
$\mathbb{F}_q$ is the finite field with $q$ elements, $p$ is its characteristic,
and $n\geq 2$ is a fixed integer, which is the length of the codes under consideration.
We consider $q$-ary quantum error-correcting codes (QECCs)
and entanglement-assisted QECCs (EAQECCs).

Our quantum codes are defined by $\mathbb{F}_q$-linear spaces $C\subseteq\mathbb{F}_q^{2n}$,
whose vectors have one pair of coordinates for each of the $n$ qudits of a codeword.
Thus $n$ is the length of $C$ and of the quantum code defined by $C$.
In Subsection \ref{sec212} we also use an auxiliary code of length $n+c$,
where $c$ is the number of maximally entangled pairs of the entanglement-assisted code given by $C$
(see (\ref{def:c})).
Since the notions below are applied to both lengths,
we state them for an arbitrary positive integer $N$, which in this paper is either $n$ or $n+c$.
We write the vectors of $\mathbb{F}_q^{2N}$ as
$\mathbf{x}=(a_1,\ldots,a_N\mid b_1,\ldots,b_N)$,
and we say that $(a_i,b_i)$ are the coordinates of $\mathbf{x}$ at position (or qudit) $i$.

\begin{definition}\label{def:symplectic}
Let $\mathbf{x}=(a_1,\ldots,a_N\mid b_1,\ldots,b_N)$ and
$\mathbf{y}=(a'_1,\ldots,a'_N\mid b'_1,\ldots,b'_N)$ be vectors in $\mathbb{F}_q^{2N}$.
\begin{enumerate}
\item The \emph{symplectic inner product} of $\mathbf{x}$ and $\mathbf{y}$ is
\begin{equation*}
  \langle \mathbf{x}, \mathbf{y} \rangle_s = \sum_{i=1}^N \left(a_i b'_i - b_i a'_i\right).
\end{equation*}
\item The \emph{support} of $\mathbf{x}$ is
$\operatorname{supp}(\mathbf{x}) = \{ i : (a_i, b_i) \neq (0,0) \}\subseteq\{1,\ldots,N\}$,
and the \emph{symplectic weight} of $\mathbf{x}$ is
$w_s(\mathbf{x}) = |\operatorname{supp}(\mathbf{x})|$.
\item For a nonempty set $\mathcal{X}\subseteq\mathbb{F}_q^{2N}$,
the support $\operatorname{supp}(\mathcal{X})$ is the union of the supports of the vectors in $\mathcal{X}$,
and, if $\mathcal{X}\not\subseteq\{\mathbf{0}\}$,
$w_s(\mathcal{X})$ denotes the minimum symplectic weight of the nonzero vectors in $\mathcal{X}$.
\item For an $\mathbb{F}_q$-linear space $A\subseteq\mathbb{F}_q^{2N}$, its \emph{symplectic dual} is
$A^{\perp s}=\bigl\{\mathbf{y}\in\mathbb{F}_q^{2N} : \langle \mathbf{x}, \mathbf{y} \rangle_s=0
\text{ for all } \mathbf{x}\in A\bigr\}$.
\end{enumerate}
\end{definition}

Since the symplectic form is nondegenerate, $(A^{\perp s})^{\perp s}=A$
and $\dim A+\dim A^{\perp s}=2N$ for every $\mathbb{F}_q$-linear space $A\subseteq\mathbb{F}_q^{2N}$.

\begin{definition}\label{def:puncture}
Let $K\subseteq\{1,\ldots,N\}$ be a nonempty set.
\begin{enumerate}
\item For $\mathbf{a}=(a_1,\ldots,a_N)\in\mathbb{F}_q^N$, the \emph{projection} of $\mathbf{a}$ onto $K$ is
$\pi_K(\mathbf{a})=(a_i)_{i\in K}\in\mathbb{F}_q^{|K|}$.
For $(\mathbf{a}\mid\mathbf{b})\in\mathbb{F}_q^{2N}$ we set
$\pi_K(\mathbf{a}\mid\mathbf{b})=(\pi_K(\mathbf{a})\mid\pi_K(\mathbf{b}))\in\mathbb{F}_q^{2|K|}$.
\item For an $\mathbb{F}_q$-linear space $A\subseteq\mathbb{F}_q^{N}$,
the \emph{punctured code} and the \emph{shortened code} of $A$ onto $K$ are, respectively,
\begin{align*}
  \pi_K(A) & = \bigl\{ \pi_K(\mathbf{a}) : \mathbf{a} \in A \bigr\}, \\
  \sigma_K(A) & = \bigl\{ \pi_K(\mathbf{a}) : \mathbf{a}=(a_1,\ldots,a_N) \in A
  \text{ and } a_i = 0 \text{ for all } i \notin K \bigr\}.
\end{align*}
\item For an $\mathbb{F}_q$-linear space $A\subseteq\mathbb{F}_q^{2N}$,
the \emph{punctured code} and the \emph{shortened code} of $A$ onto $K$ are, respectively,
\begin{align*}
  \pi_K(A) & = \bigl\{ \pi_K(\mathbf{x}) : \mathbf{x} \in A \bigr\}, \\
  \sigma_K(A) & = \bigl\{ \pi_K(\mathbf{x}) : \mathbf{x} \in A
  \text{ and } \operatorname{supp}(\mathbf{x})\subseteq K \bigr\}.
\end{align*}
\end{enumerate}
\end{definition}

For $K=\emptyset$ we set $\pi_\emptyset(A)=\sigma_\emptyset(A)=\{\mathbf{0}\}$.
Thus $\pi_K(A)$ and $\sigma_K(A)$ are subspaces of $\mathbb{F}_q^{|K|}$
(respectively, of $\mathbb{F}_q^{2|K|}$) whose coordinates are indexed by $K$.
When $A$ is a subspace of $\mathbb{F}_q^{2|K|}$ indexed by $K$, for instance $A=\pi_K(B)$,
and $I\subseteq K$, the spaces $\pi_I(A)$ and $\sigma_I(A)$
are defined in the same way, with $\{1,\ldots,N\}$ replaced by $K$.
Punctured codes always lie in their own ambient space.
For instance, $\pi_K(A)\subseteq\mathbb{F}_q^{2|K|}$.
In contrast, a shortened code $\sigma_K(A)$, with $A\subseteq\mathbb{F}_q^{2N}$,
will sometimes be regarded as a subspace of $\mathbb{F}_q^{2N}$
by padding its vectors with zeros at the positions outside $K$.
In this way $\sigma_K(A)$ is identified with
$\{\mathbf{x}\in A : \operatorname{supp}(\mathbf{x})\subseteq K\}$.
This identification is used only when it is explicitly mentioned or clear from the context
(for example, in a direct sum decomposition of a subspace of $\mathbb{F}_q^{2N}$).

Puncturing and shortening are dual operations with respect to the symplectic inner product,
as stated in the following result.

\begin{proposition}[{\citep{galindo19}}] \label{prop:dual}
Let $C\subseteq\mathbb{F}_q^{2N}$ be an $\mathbb{F}_q$-linear space
and let $K\subseteq\{1,\ldots,N\}$ be nonempty.
Then, with symplectic duals taken in $\mathbb{F}_q^{2|K|}$,
\begin{equation}\label{eq:dual}
  \begin{split}
  \pi_K(C)^{\perp s} &= \sigma_K(C^{\perp s}),\\
  \sigma_K(C)^{\perp s} &= \pi_K(C^{\perp s}).
  \end{split}
\end{equation}
\end{proposition}

\subsubsection{Entanglement-assisted quantum error-correcting codes}\label{sec212}
Let $C\subseteq\mathbb{F}_q^{2n}$ be an $\mathbb{F}_q$-linear space.
Following \citep{brun06,luo17,galindo19}, the space $C$ defines an EAQECC, $Q(C)$.
Throughout the paper we assume that $C \supsetneq C \cap C^{\perp s}$,
that is, $C\not\subseteq C^{\perp s}$.
This guarantees that the minimum distance in Definition \ref{def:QC} is well defined.

By \citep[Theorem 2 and Proposition 2]{galindo19} the \emph{number of maximally entangled pairs} of $C$ is
\begin{equation}
\label{def:c}
  c=\frac{1}{2}\left(\dim C^{\perp s} - \dim (C \cap C^{\perp s})\right).
\end{equation}

Note that our standing assumption $C\not\subseteq C^{\perp s}$ does not imply $c>0$.
Indeed, $c=0$ if and only if $C^{\perp s}\subseteq C$,
that is, if and only if $Q(C)$ is an entanglement-unassisted quantum stabilizer code.
Unless otherwise stated, the results below hold for $c\ge0$.
When $c=0$ we take $C'=C$; then the maps $\alpha$ and $\beta$ defined below are the zero map and the identity, respectively.

The code $Q(C)$ is described by means of an auxiliary quantum stabilizer code
of length $n+c$, which is defined by the space $C'$ of the following lemma.

\begin{lemma}\label{lem:Cprime}
There exists an $\mathbb{F}_q$-linear space $C'\subseteq\mathbb{F}_q^{2(n+c)}$ such that
\begin{enumerate}
\item\label{cp1} $(C')^{\perp s}\subseteq C'$,
\item\label{cp2} $\sigma_{\{1,\ldots,n\}}(C')=C$, and
\item\label{cp3} $\pi_{\{1,\ldots,n\}}((C')^{\perp s})=C^{\perp s}$.
\end{enumerate}
Moreover, for every such $C'$, $\pi_{\{1,\ldots,n\}}$ maps $(C')^{\perp s}$ bijectively onto $C^{\perp s}$,
and $\dim C'=\dim C+2c$.
\end{lemma}
\begin{proof}
The radical of the restriction of the symplectic form to $C^{\perp s}$ is
$C^{\perp s}\cap (C^{\perp s})^{\perp s}=C\cap C^{\perp s}$.
Hence the symplectic form induces a nondegenerate alternating form on the quotient
$C^{\perp s}/(C\cap C^{\perp s})$, whose dimension is $2c$.
Therefore there exist
$\mathbf{u}_1,\mathbf{v}_1,\ldots,\mathbf{u}_c,\mathbf{v}_c\in C^{\perp s}$ that complete a basis of
$C\cap C^{\perp s}$ to a basis of $C^{\perp s}$ and satisfy
$\langle \mathbf{u}_i,\mathbf{v}_j\rangle_s=\delta_{ij}$ and
$\langle \mathbf{u}_i,\mathbf{u}_j\rangle_s=\langle \mathbf{v}_i,\mathbf{v}_j\rangle_s=0$
for $1\le i,j\le c$, where $\delta_{ij}$ is the Kronecker delta.
We extend these vectors to vectors in $\mathbb{F}_q^{2(n+c)}$ as follows:
every vector of $C\cap C^{\perp s}$ is extended by zeros,
$\mathbf{u}_j$ is extended by setting $a_{n+j}=1$,
and $\mathbf{v}_j$ by setting $b_{n+j}=-1$,
all the other new coordinates being $0$.
Let $U$ be the $\mathbb{F}_q$-span of the extended vectors
and put $C'=U^{\perp s}$.
The extended vectors satisfy
$\langle \mathbf{u}_i,\mathbf{v}_j\rangle_s=\delta_{ij}-\delta_{ij}=0$,
and all the other symplectic inner products among them are also zero.
Hence $U\subseteq U^{\perp s}=C'$, and $(C')^{\perp s}=U$, which proves Item (\ref{cp1}).
By construction, $\pi_{\{1,\ldots,n\}}$ maps a basis of $U$ onto a basis of $C^{\perp s}$.
Therefore it maps $(C')^{\perp s}=U$ bijectively onto $C^{\perp s}$, which proves Item (\ref{cp3}).
By Proposition \ref{prop:dual},
$\sigma_{\{1,\ldots,n\}}(C')=\pi_{\{1,\ldots,n\}}((C')^{\perp s})^{\perp s}=(C^{\perp s})^{\perp s}=C$,
which proves Item (\ref{cp2}).

The last two assertions hold for every $C'$ satisfying Items (\ref{cp1})--(\ref{cp3}).
Indeed, let $Z=\sigma_{\{1,\ldots,n\}}((C')^{\perp s})$.
By Items (\ref{cp1}), (\ref{cp2}) and (\ref{cp3}), $Z\subseteq\sigma_{\{1,\ldots,n\}}(C')\cap\pi_{\{1,\ldots,n\}}((C')^{\perp s})=C\cap C^{\perp s}$,
so that
\begin{equation*}
\dim (C')^{\perp s}=\dim Z+\dim\pi_{\{n+1,\ldots,n+c\}}((C')^{\perp s})\le\dim(C\cap C^{\perp s})+2c=\dim C^{\perp s}.
\end{equation*}
Since $\pi_{\{1,\ldots,n\}}$ maps $(C')^{\perp s}$ onto $C^{\perp s}$ by Item (\ref{cp3}), it is bijective,
$Z=C\cap C^{\perp s}$, and $\pi_{\{n+1,\ldots,n+c\}}((C')^{\perp s})=\mathbb{F}_q^{2c}$.
Hence $\dim C'=2(n+c)-\dim C^{\perp s}=\dim C+2c$.
\end{proof}

From now on, $C'$ denotes a fixed $\mathbb{F}_q$-linear space satisfying
the conditions in Items (\ref{cp1})--(\ref{cp3}) of Lemma \ref{lem:Cprime},
for instance, the space constructed in its proof.
The choice of $C'$ is not unique.
When a result depends on it, this will be explicitly indicated.

\emph{The map $\alpha$.}
In \citep[Proof of Theorem 3.4]{eaqlrc26} it is shown that there exists an $\mathbb{F}_p$-linear epimorphism
\begin{equation*}
  \alpha : C^{\perp s} \rightarrow \mathbb{F}_q^{2c}
\end{equation*}
such that, according to the notation in Definition \ref{def:symplectic},
\begin{multline}\label{eq2}
  (C')^{\perp s} = \bigl\{ (a_1, \ldots, a_{n+c}\mid b_1, \ldots, b_{n+c}) :\ \mathbf{x}\in C^{\perp s} \text{ and}\\
  (a_{n+1}, \ldots, a_{n+c}\mid b_{n+1},\ldots, b_{n+c})= \alpha(\mathbf{x}) \bigr\},
\end{multline}
where $\mathbf{x}= (a_1, \ldots, a_n\mid b_1,\ldots, b_n)$.
In other words, $(C')^{\perp s}$ is the graph of $\alpha$:
for $\mathbf{x}\in C^{\perp s}$, the vector $\alpha(\mathbf{x})$ consists of the coordinates
at positions $n+1,\ldots,n+c$ of the unique vector of $(C')^{\perp s}$
whose projection onto $\{1,\ldots,n\}$ is $\mathbf{x}$.
Since $C'$, and hence $(C')^{\perp s}$, is $\mathbb{F}_q$-linear,
the map $\alpha$ is $\mathbb{F}_q$-linear.
Physically, $\mathbf{x}\in C^{\perp s}$ describes the action of a stabilizer element
on the $n$ transmitted qudits, and $\alpha(\mathbf{x})$ describes its action
on the $c$ qudits held by the decoder,
which are the halves of the maximally entangled pairs shared with the encoder.
For the space $C'$ constructed in the proof of Lemma \ref{lem:Cprime},
these properties can be checked directly:
$\alpha$ vanishes on $C\cap C^{\perp s}$,
and $\alpha(\mathbf{u}_j)$ and $-\alpha(\mathbf{v}_j)$ are the $j$th and the $(c+j)$th
vectors of the standard basis of $\mathbb{F}_q^{2c}$.

\emph{The map $\beta$.}
For $\mathbf{x}=(a_1,\ldots,a_n\mid b_1,\ldots,b_n)\in C^{\perp s}$, define
\begin{equation*}
  \beta(\mathbf{x}) = (a_1, \ldots, a_{n+c}\mid b_1, \ldots, b_{n+c}),
\end{equation*}
where $(a_{n+1}, \ldots, a_{n+c}\mid b_{n+1},\ldots, b_{n+c})= \alpha(\mathbf{x})$.
By (\ref{eq2}), $\beta$ is an $\mathbb{F}_q$-linear isomorphism
from $C^{\perp s}$ onto $(C')^{\perp s}$, whose inverse is the restriction of
$\pi_{\{1,\ldots,n\}}$ to $(C')^{\perp s}$.
In particular, $\pi_T(\beta(\mathbf{x}))=\pi_T(\mathbf{x})$
for every $\mathbf{x}\in C^{\perp s}$ and every nonempty $T\subseteq\{1,\ldots,n\}$.
The map $\beta$ lifts a stabilizer vector acting on the $n$ transmitted qudits
to the corresponding stabilizer vector of the auxiliary code acting on all the $n+c$ qudits.
Particular constructions of $\alpha$ and $\beta$, together with the properties recalled here,
were given by \citet{brun06} and \citet{luo17}
in terms of complex unitary matrices when $q$ is a prime number.


We can now introduce the two quantum codes associated with $C$.

\begin{definition}\label{def:QC}\leavevmode
\csname @beginparpenalty\endcsname=10000
\begin{enumerate}[wide, labelindent=0pt]
\item $Q(C')$\label{QCp} denotes the $q$-ary quantum stabilizer code
\citep{ashikhmin00,calderbank98,gottesman96,ketkar06}
whose stabilizer is defined by $(C')^{\perp s}$ and whose normalizer is defined by $C'$,
that is, a common eigenspace of the commuting Weyl operators defined by the vectors of $(C')^{\perp s}$.
The eigenvalues are fixed and known to the decoder, and their choice plays no role in this paper,
because the measurement outcomes are read relative to them (Subsection \ref{sec213}).
Its length is $n+c$ and it encodes $\dim C' - (n+c)=\dim C+c-n$ qudits,
that is, its dimension as a complex vector space is $q^{\dim C+c-n}$.
Its minimum distance is $w_s(C'\setminus (C')^{\perp s})$.
Thus $Q(C')$ is an $[[n+c,\ \dim C+c-n,\ w_s(C'\setminus (C')^{\perp s})]]_q$ code.
\item\label{QC} $Q(C)$ denotes the EAQECC defined by $C$
\citep{brun06,luo17,galindo19}.
Every codeword of $Q(C)$ is the partial trace of a codeword of $Q(C')$
over the $c$ qudits held by the decoder \citep{galindo19,eaqlrc26}.
The code $Q(C)$ encodes $c+\dim C - n$ information qudits
into $n$ codeword qudits using $c$ pairs of maximally entangled quantum states,
and its minimum distance is $d=w_s(C \setminus (C \cap C^{\perp s}))$
\citep[Theorem 2]{galindo19}.
Thus $Q(C)$ is an $[[n,\ c+\dim C - n,\ d;\ c]]_q$ EAQECC.
\end{enumerate}
\end{definition}

The number of encoded qudits is positive.
Indeed, since $\dim C^{\perp s}=2n-\dim C$, (\ref{def:c}) gives
\begin{equation*}
c+\dim C-n=\tfrac{1}{2}\left(\dim C-\dim(C\cap C^{\perp s})\right)>0.
\end{equation*}
The EAQECC $Q(C)$ can simultaneously correct $e$ erasures and $t$ errors
whenever $e + 2t < d$.

\subsubsection{Decoding and measured qudits}\label{sec213}
A decoder for $Q(C)$ concatenates the received quantum state of $n$ qudits
with the $c$ halves of the maximally entangled pairs that it shares with the encoder,
thus forming a state of $n+c$ qudits,
and applies a decoder for $Q(C')$ to these $n+c$ qudits.
Hence the decoder acts jointly on the received register and on its own correlated register.
The received state alone is the partial trace over the decoder's register.

The decoding procedure described above is the original one proposed by \citet{brun06},
and it reconstructs codewords of $Q(C')$. Entanglement-assisted quantum local recovery has been studied in
\citep{EAQLR-gret,EAQLR-r,EAQLR-r2,eaqlrc26}.

In this paper we consider the following model for the measurements performed by a decoder,
which follows \citep[Assumption 1 and Remark 2]{qlrsurface25eprint}
and which we call the \emph{stabilizer measurement model}:
the decoder measures observables defined by vectors of the stabilizer,
and then applies a Weyl (generalized Pauli) correction that depends only on the measurement outcomes,
as in the decoding procedure of \citep[Section 3.2]{qlrc24}, which is used in the proof of
\citep[Theorem 3]{qlrsurface25eprint}.
More precisely, a basis $B = \{ \mathbf{x}_1, \ldots, \mathbf{x}_r \}$ of an $\mathbb{F}_q$-subspace
of $C^{\perp s}$ induces the set of observables defined by
$\beta(\mathbf{x}_1), \ldots, \beta(\mathbf{x}_r)$,
which act on the $n+c$ qudits held by the decoder after the reception of
a transmitted codeword of $n$ qudits.
Assumption 1 and the theorems of \citep{qlrsurface25eprint} are stated for prime $q$,
and \citep[Remark 2]{qlrsurface25eprint} points out that they remain valid for $q=p^m$,
the number of measurements being proportional to $|B|$.
Indeed, let $\eta_1,\ldots,\eta_m$ be a basis of $\mathbb{F}_q$ over $\mathbb{F}_p$,
let $\operatorname{tr}$ be the trace map from $\mathbb{F}_q$ to $\mathbb{F}_p$,
and let $\zeta$ be a primitive $p$th root of unity.
We fix the following convention for Weyl operators.
Let $\{\ket{u} : u\in\mathbb{F}_q\}$ be an orthonormal basis of $\mathbb{C}^q$ and, for $a,b\in\mathbb{F}_q$,
let $X(a)$ and $Z(b)$ be the operators on $\mathbb{C}^q$ given by
$X(a)\ket{u}=\ket{u+a}$ and $Z(b)\ket{u}=\zeta^{\operatorname{tr}(bu)}\ket{u}$.
For $\mathbf{x}=(a_1,\ldots,a_N\mid b_1,\ldots,b_N)\in\mathbb{F}_q^{2N}$,
the Weyl operator defined by $\mathbf{x}$ is
$E_{\mathbf{x}}=X(a_1)Z(b_1)\otimes\cdots\otimes X(a_N)Z(b_N)$.
Since $Z(b)X(a)=\zeta^{\operatorname{tr}(ab)}X(a)Z(b)$, for all $\mathbf{x},\mathbf{y}\in\mathbb{F}_q^{2N}$ we have
\begin{equation*}
  E_{\mathbf{x}}E_{\mathbf{y}}=\zeta^{-\operatorname{tr}\langle\mathbf{x},\mathbf{y}\rangle_s}E_{\mathbf{y}}E_{\mathbf{x}},
\end{equation*}
with the order of the arguments of Definition \ref{def:symplectic}.
A vector $\mathbf{x}$ of $(C')^{\perp s}$ gives the $m$ commuting $p$-ary Pauli observables
defined by $\eta_1\mathbf{x},\ldots,\eta_m\mathbf{x}$.
For $1\le j\le m$, let $\lambda_j$ be the eigenvalue of $E_{\eta_j\mathbf{x}}$ on the codewords of $Q(C')$
(Definition \ref{def:QC}), which is known to the decoder.
If a codeword $\ket{\psi}$ of $Q(C')$ is affected by the error $E_{\mathbf{e}}$, with $\mathbf{e}\in\mathbb{F}_q^{2(n+c)}$,
the above commutation relation gives
\begin{equation*}
  E_{\eta_j\mathbf{x}}E_{\mathbf{e}}\ket{\psi}
  =\lambda_j\,\zeta^{-\operatorname{tr}(\eta_j\langle\mathbf{x},\mathbf{e}\rangle_s)}E_{\mathbf{e}}\ket{\psi}.
\end{equation*}
Hence the measurement of the $j$th observable returns the eigenvalue $\mu_j=\lambda_j\zeta^{-s_j}$, where
\begin{equation*}
  s_j=\operatorname{tr}(\eta_j\langle\mathbf{x},\mathbf{e}\rangle_s)\in\mathbb{F}_p,\qquad 1\le j\le m,
\end{equation*}
and the decoder obtains $s_j$ from the ratio $\mu_j/\lambda_j=\zeta^{-s_j}$.
We call $s_1,\ldots,s_m$ the outcomes of these observables on the error $\mathbf{e}$.
These determine $\langle\mathbf{x},\mathbf{e}\rangle_s\in\mathbb{F}_q$ because the trace form is nondegenerate
(see also the proof of \citep[Theorem 1]{EAQLR-r}).
Thus, when $q=p^m$, one vector of $\mathbb{F}_q^{2(n+c)}$ corresponds to $m$ elementary
$p$-ary Pauli measurements, or to one measurement with $q$ outcomes.
Throughout the paper, we count observables by the vectors that define them.
The set of measured qudits is
\begin{equation*}
  \bigcup_{i=1}^r \operatorname{supp}(\beta(\mathbf{x}_i))
  \subseteq\{1,\ldots,n+c\}.
\end{equation*}
Since $\pi_{\{1,\ldots,n\}}(\beta(\mathbf{x}_i))=\mathbf{x}_i$,
the set of measured codeword qudits is
\begin{equation*}
  \bigcup_{i=1}^r \operatorname{supp}(\mathbf{x}_i)\subseteq\{1,\ldots,n\}.
\end{equation*}
A qudit $n+j$ held by the decoder, $1\le j\le c$, is measured
if and only if some $\beta(\mathbf{x}_i)$ is nonzero at position $n+j$.

Recovery procedures outside this model may use fewer measurements, or none.
For instance, if every codeword of $Q(C')$ has the form $\ket{\tau}\otimes\ket{\phi}$,
where $\ket{\tau}$ is a fixed pure state of the erased qudits
(as in the proof of Proposition \ref{prop:KI}),
the erased qudits can be replaced by fresh qudits prepared in the state $\ket{\tau}$
without any measurement.
Purity is essential. Indeed, the reduced state of the erased qudits is the same for all codewords
whenever the erasures are correctable, but when it is mixed,
replacing the erased qudits by it destroys their correlations with the remaining qudits.
Hence, the minimality statements of this paper on the number of observables
refer to the stabilizer measurement model.


\subsubsection{Entanglement-assisted quantum local recovery}\label{sec214}
We now recall when erasures can be corrected by acting on a small set of qudits.
The following definition is taken from \citep[Definition 3.1]{eaqlrc26},
where the inclusion $I\subsetneq K$ is required.
We also allow $K=I$, a case that is justified by Proposition \ref{prop:KI} below.

\begin{definition}\label{def:repair}
Let $\emptyset \neq I \subseteq K \subseteq \{1, \ldots, n\}$.
Assume that the qudits at the positions in $I$ of a codeword of $Q(C)$ are erased,
that is, subjected to the completely depolarizing channel,
while the $c$ qudits held by the decoder are not affected.
The set $K$ is called a \emph{repair group} for $I$
(and the erasures at $I$ are said to be \emph{locally recoverable} from $K$)
if there exists a quantum operation acting only on the qudits in $K\cup\{n+1,\ldots,n+c\}$
that restores every codeword of $Q(C')$.
\end{definition}

When $c=0$, Definition \ref{def:repair} reduces to the usual notion
for quantum stabilizer codes \citep{qlrc24}, with $C'=C$.

\begin{theorem}\label{thm:EAlocal}
Let $\emptyset \neq I \subsetneq K \subseteq \{1, \ldots, n\}$.
Then $K$ is a repair group for $I$ if and only if
\begin{equation}\label{eq1}
  \sigma_I(\pi_K(C)) = \sigma_I(C \cap C^{\perp s}),
\end{equation}
or, equivalently,
\begin{equation}\label{eq1b}
  \pi_I(C+C^{\perp s}) = \pi_I (\sigma_K (C^{\perp s})).
\end{equation}
\end{theorem}

For $c>0$, Theorem \ref{thm:EAlocal} is \citep[Theorem 3.4]{eaqlrc26},
where it is stated with the roles of $C$ and $C^{\perp s}$ interchanged.
For $c=0$, that is, $C\supseteq C^{\perp s}$,
we have $C+C^{\perp s}=C$ and $C\cap C^{\perp s}=C^{\perp s}$,
and the criterion becomes $\pi_I(C)=\pi_I(\sigma_K(C^{\perp s}))$,
which was established in \citep{qlrc24}.
Conditions (\ref{eq1}) and (\ref{eq1b}) are equivalent by Proposition \ref{prop:dual}:
taking symplectic duals in $\mathbb{F}_q^{2|I|}$, we have
$$\sigma_I(\pi_K(C))^{\perp s}=\pi_I(\pi_K(C)^{\perp s})=\pi_I(\sigma_K(C^{\perp s}))$$
and
$\sigma_I(C\cap C^{\perp s})^{\perp s}=\pi_I((C\cap C^{\perp s})^{\perp s})=\pi_I(C+C^{\perp s})$.

For completeness, we show that the case $c>0$ of Theorem \ref{thm:EAlocal}
can be derived from the stabilizer case \citep{qlrc24}.
Put $L=K\cup\{n+1,\ldots,n+c\}$.
By Definition \ref{def:repair}, acting on $K$ in $Q(C)$ means acting on $L$ in $Q(C')$.
Since $C'\supseteq(C')^{\perp s}$ and $I\subsetneq L$,
the criterion of \citep{qlrc24} applied to $Q(C')$ states that
$L$ is a repair group for $I$ in $Q(C')$ if and only if
$\pi_I(C')=\pi_I(\sigma_L((C')^{\perp s}))$.
The following proposition shows that this condition coincides with (\ref{eq1b}).

\begin{proposition}\label{prop:reduction}
Assume $c>0$, let $\emptyset \neq I \subseteq K \subseteq \{1, \ldots, n\}$
and put $L=K\cup\{n+1,\ldots,n+c\}$.
Then
\begin{enumerate}
\item\label{red1} $\ker\alpha=C\cap C^{\perp s}$ and
$\pi_{\{1,\ldots,n\}}(C')=C+C^{\perp s}$;
\item\label{red2} $\pi_I(C')=\pi_I(\sigma_L((C')^{\perp s}))$ if and only if
$\pi_I(C+C^{\perp s}) = \pi_I (\sigma_K (C^{\perp s}))$.
\end{enumerate}
\end{proposition}
\begin{proof}
We begin by proving Item (\ref{red1}).
By (\ref{eq2}), $\ker\alpha$ consists of the vectors $\mathbf{x}\in C^{\perp s}$
such that $\beta(\mathbf{x})$ is zero at the positions $n+1,\ldots,n+c$,
that is, $\ker\alpha=\sigma_{\{1,\ldots,n\}}((C')^{\perp s})$.
Since $(C')^{\perp s}\subseteq C'$, we obtain
$\ker\alpha\subseteq\sigma_{\{1,\ldots,n\}}(C')\cap C^{\perp s}=C\cap C^{\perp s}$.
As $\alpha$ is onto $\mathbb{F}_q^{2c}$,
$\dim\ker\alpha=\dim C^{\perp s}-2c=\dim (C\cap C^{\perp s})$ by (\ref{def:c}).
Hence $\ker\alpha=C\cap C^{\perp s}$.
Applying Proposition \ref{prop:dual} to $(C')^{\perp s}$, we obtain
$$\pi_{\{1,\ldots,n\}}(C')=\sigma_{\{1,\ldots,n\}}((C')^{\perp s})^{\perp s}
=(C\cap C^{\perp s})^{\perp s}=C+C^{\perp s}.$$

Now, let us prove Item (\ref{red2}).
By Item (\ref{red1}), $\pi_I(C')=\pi_I(\pi_{\{1,\ldots,n\}}(C'))=\pi_I(C+C^{\perp s})$.
On the other hand, a vector $\beta(\mathbf{x})\in(C')^{\perp s}$, with $\mathbf{x}\in C^{\perp s}$,
is supported on $L$ if and only if $\operatorname{supp}(\mathbf{x})\subseteq K$.
Hence the vectors of $(C')^{\perp s}$ supported on $L$ are exactly the vectors $\beta(\mathbf{x})$
with $\mathbf{x}$ in $C^{\perp s}$ and supported on $K$.
Since $\pi_I(\beta(\mathbf{x}))=\pi_I(\mathbf{x})$, we conclude that
$\pi_I(\sigma_L((C')^{\perp s}))=\pi_I(\sigma_K(C^{\perp s}))$,
and the result follows.
\end{proof}

Theorem \ref{thm:EAlocal} is stated in \citep{qlrc24,eaqlrc26}
under the assumption $I\subsetneq K$.
The following proposition shows that the criterion remains valid when $K=I$.

\begin{proposition}\label{prop:KI}
Let $\emptyset\neq I\subseteq\{1,\ldots,n\}$.
Then $I$ is a repair group for $I$ if and only if (\ref{eq1}) holds with $K=I$.
Consequently, Theorem \ref{thm:EAlocal} holds for every
$\emptyset \neq I \subseteq K \subseteq \{1, \ldots, n\}$.
\end{proposition}
\begin{proof}
Assume first that $c>0$ and set $L=I\cup\{n+1,\ldots,n+c\}$, so that $I\subsetneq L$.
By Definition \ref{def:repair}, $I$ is a repair group for $I$ in $Q(C)$ if and only if
$L$ is a repair group for $I$ in the stabilizer code $Q(C')$,
that is \citep{qlrc24}, if and only if $\pi_I(C')=\pi_I(\sigma_L((C')^{\perp s}))$.
By Item (\ref{red2}) of Proposition \ref{prop:reduction} with $K=I$,
this is (\ref{eq1b}), and hence (\ref{eq1}), with $K=I$.

Assume now that $c=0$, that is, $C\supseteq C^{\perp s}$.
Then $C\cap C^{\perp s}=C^{\perp s}$, $\sigma_I(\pi_I(C))=\pi_I(C)$,
and (\ref{eq1}) with $K=I$ reads $\pi_I(C)=\sigma_I(C^{\perp s})$.
By Proposition \ref{prop:dual}, $\pi_I(C)^{\perp s}=\sigma_I(C^{\perp s})$.
Hence (\ref{eq1}) with $K=I$ is equivalent to
\begin{equation}\label{eq:selfdual}
  \pi_I(C)= \pi_I(C)^{\perp s}=\sigma_I(C^{\perp s}).
\end{equation}
In the following, for $\mathbf{x}\in\mathbb{F}_q^{2N}$ we denote by $E_{\mathbf{x}}$
the Weyl (generalized Pauli) operator defined by $\mathbf{x}$ \citep{ashikhmin00,ketkar06},
with the convention fixed in Subsection \ref{sec213}.
For $\mathbf{x},\mathbf{y}\in\mathbb{F}_q^{2N}$, the operators $E_{\mathbf{x}}$ and $E_{\mathbf{y}}$
commute up to the phase $\zeta^{-\operatorname{tr}\langle\mathbf{x},\mathbf{y}\rangle_s}$,
that is, $E_{\mathbf{x}}E_{\mathbf{y}}=\zeta^{-\operatorname{tr}\langle\mathbf{x},\mathbf{y}\rangle_s}E_{\mathbf{y}}E_{\mathbf{x}}$,
where $\zeta$ and $\operatorname{tr}$ are as in Subsection \ref{sec213}.

Suppose that (\ref{eq:selfdual}) holds.
The space $\sigma_I(C^{\perp s})$ has dimension $|I|$,
and the operators defined by its vectors act only on the qudits in $I$
and stabilize every codeword of $Q(C)$ with fixed eigenvalues.
Since $\sigma_I(C^{\perp s})$ is self-dual in $\mathbb{F}_q^{2|I|}$,
the corresponding common eigenspace on the qudits in $I$ is one-dimensional,
spanned by a pure state $\ket{\tau}$.
Hence every codeword of $Q(C)$ has the form $\ket{\tau}\otimes\ket{\phi}$,
where $\ket{\tau}$ is the same for all codewords.
Therefore the erasures at $I$ are corrected by acting only on the qudits in $I$:
it suffices to replace the erased qudits by $\ket{\tau}$,
which can be done, for instance, by the stabilizer measurement of $\sigma_I(C^{\perp s})$
followed by a suitable Weyl correction on $I$.
Thus $I$ is a repair group for $I$.

Conversely, suppose that $I$ is a repair group for $I$,
and let $\mathcal{R}$ be a quantum operation acting only on the qudits in $I$
that restores every codeword of $Q(C)$.
Let $\rho$ be the density matrix of a codeword $\ket{\psi}$ of $Q(C)$.
After the erasures, the state is $\rho_{\bar I}\otimes \mathbf{1}/q^{|I|}$,
where $\rho_{\bar I}$ is the reduced density matrix on the complement $\bar I$ of $I$,
that is, the partial trace of $\rho$ over the qudits in $I$,
and $\mathbf{1}/q^{|I|}$ is the completely mixed state on $I$.
After the recovery, the state is $\rho_{\bar I}\otimes\tau_0$, where
$\tau_0=\mathcal{R}(\mathbf{1}/q^{|I|})$ does not depend on the codeword.
Since $\rho_{\bar I}\otimes\tau_0=\rho$ is pure, both factors are pure.
Hence every codeword has the form $\ket{\tau}\otimes\ket{\phi}$,
with the same pure state $\ket{\tau}$ on $I$.

Let $\mathbf{x}\in\pi_I(C)$ and take $\mathbf{u}\in C$ with $\pi_I(\mathbf{u})=\mathbf{x}$.
Up to a reordering of the tensor factors,
$E_{\mathbf{u}}=E_{\mathbf{x}}\otimes E_{\mathbf{u}'}$, where $\mathbf{u}'$ is the projection of
$\mathbf{u}$ onto $\bar I$.
Since $\mathbf{u}$ belongs to $C$, which defines the normalizer,
$E_{\mathbf{u}}$ is a logical operator.
Thus it maps the codeword $\ket{\tau}\otimes\ket{\phi}$ to a codeword $\ket{\tau}\otimes\ket{\phi'}$,
and consequently $E_{\mathbf{x}}\ket{\tau}$ is proportional to $\ket{\tau}$.
Now let $\mathbf{x},\mathbf{y}\in\pi_I(C)$.
Both $E_{\mathbf{x}}$ and $E_{\mathbf{y}}$ have $\ket{\tau}$ as an eigenvector,
which implies that the commutation phase $\zeta^{-\operatorname{tr}\langle\mathbf{x},\mathbf{y}\rangle_s}$ equals $1$.
As $\pi_I(C)$ is $\mathbb{F}_q$-linear, the same holds for $\lambda\mathbf{x}$ and $\mathbf{y}$,
and consequently $\operatorname{tr}(\lambda\langle\mathbf{x},\mathbf{y}\rangle_s)=0$ for every $\lambda\in\mathbb{F}_q$.
Since the trace form is nondegenerate, $\langle\mathbf{x},\mathbf{y}\rangle_s=0$.
Therefore $$\pi_I(C)\subseteq\pi_I(C)^{\perp s}.$$
The reverse inclusion always holds, since $C^{\perp s}\subseteq C$ gives
$$\pi_I(C)^{\perp s}=\sigma_I(C^{\perp s})\subseteq\pi_I(C^{\perp s})\subseteq\pi_I(C).$$
Hence (\ref{eq:selfdual}) holds, which is (\ref{eq1}) with $K=I$.
\end{proof}

We conclude this subsection by recalling the results of \citep{qlrsurface25eprint}
on the \emph{minimum number of observables needed for correcting erasures
in the stabilizer measurement model} (Subsection \ref{sec213}).
We state them in our notation:
in \citep{qlrsurface25eprint} a quantum stabilizer code is described by
the self-orthogonal space defining its stabilizer, which is $C^{\perp s}$ here.
For $I\subseteq\{1,\ldots,n\}$ we write $\bar I=\{1,\ldots,n\}\setminus I$,
and the shortened code $\sigma_{\bar I}(C^{\perp s})$ is regarded as a subspace of $\mathbb{F}_q^{2n}$,
that is, as the set of vectors of $C^{\perp s}$ that vanish at every position in $I$.
The observables defined by these vectors act as the identity on the erased qudits,
and they need not be measured.

\begin{theorem}[{\citep[Theorems 3 and 5, Proposition 4 and Remark 2]{qlrsurface25eprint}}]\label{thm:minobs}
Let $C^{\perp s} \subsetneq C \subsetneq \mathbb{F}_q^{2n}$,
and let $\emptyset\neq I\subseteq\{1,\ldots,n\}$ be such that $Q(C)$ can correct the erasures at $I$.
In the stabilizer measurement model, the observables defined by an $\mathbb{F}_q$-linear space
$D\subseteq C^{\perp s}$ correct the erasures at $I$ if and only if
$$C^{\perp s} = D + \sigma_{\bar I}(C^{\perp s}).$$
Consequently, the minimum number of observables for correcting the erasures at $I$ is
$\dim C^{\perp s}-\dim\sigma_{\bar I}(C^{\perp s})$,
and a space $D$ attains this minimum if and only if
\begin{equation*}
  C^{\perp s} = D \oplus \sigma_{\bar I}(C^{\perp s}).
\end{equation*}
\end{theorem}

Assume now that $c>0$ and that $Q(C)$ can correct the erasures at $I$,
for instance, because $I$ admits a repair group.
Then the quantum stabilizer code $Q(C')$ can correct the erasures at $I$,
and Theorem \ref{thm:minobs} applied to $Q(C')$ shows that
a smallest set of observables for correcting the erasures at $I$
is defined by an $\mathbb{F}_q$-linear space $D'$ such that
\begin{equation}\label{eq:smallest}
  (C')^{\perp s} = D' \oplus \sigma_{\{1, \ldots, n+c\} \setminus I}((C')^{\perp s}),
\end{equation}
where the shortened code is regarded as a subspace of $\mathbb{F}_q^{2(n+c)}$.
The next result shows that the dimension of $D'$ depends only on $C$ and $I$.

\begin{proposition}\label{prop:dimD}
Let $\emptyset\neq I\subseteq\{1,\ldots,n\}$ and let $D'$ satisfy (\ref{eq:smallest}).
Then
\begin{equation}
  \dim D' = \dim \pi_I(C^{\perp s}). \label{eq:smallest2}
\end{equation}
In particular, if $Q(C)$ can correct the erasures at $I$,
the minimum number of observables for correcting them in the stabilizer measurement model
is $\dim \pi_I(C^{\perp s})$, which does not depend on the choice of $C'$.
\end{proposition}
\begin{proof}
The kernel of the restriction of $\pi_I$ to $(C')^{\perp s}$ is
$\sigma_{\{1, \ldots, n+c\} \setminus I}((C')^{\perp s})$, regarded as a subspace of $\mathbb{F}_q^{2(n+c)}$.
Hence, by (\ref{eq:smallest}),
\begin{equation*}
  \dim D' = \dim (C')^{\perp s} - \dim \sigma_{\{1, \ldots, n+c\} \setminus I}((C')^{\perp s})
  = \dim \pi_I((C')^{\perp s}).
\end{equation*}
Since $(C')^{\perp s}=\beta(C^{\perp s})$ and $\pi_I(\beta(\mathbf{x}))=\pi_I(\mathbf{x})$
for every $\mathbf{x}\in C^{\perp s}$,
we have $\pi_I((C')^{\perp s})=\pi_I(C^{\perp s})$, which proves (\ref{eq:smallest2}).
\end{proof}

When $c=0$, we have $C'=C$ and Proposition \ref{prop:dimD} also holds,
with $D'=D$ as in Theorem \ref{thm:minobs}.

\subsection{Linear matroid parity and intersection}\label{sec2m}
In this subsection we recall two classical problems of combinatorial optimization,
which will be used in Sections \ref{sec3} and \ref{sec4} to compute reduced repair groups.
We refer to \citep{oxley11} for matroid theory
and to \citep{lovasz86,schrijver03} for both problems.
For a matrix $B$ and a set $T$ of column indices of $B$,
$B[T]$ denotes the submatrix of $B$ formed by the columns indexed by $T$.

\begin{definition}\label{def:matroid}
Let $B$ be a matrix over $\mathbb{F}_q$.
The \emph{linear matroid} (or vector matroid) of $B$ is the matroid whose ground set is the set of column indices of $B$
and in which a set of column indices is independent
if and only if the columns indexed by it are linearly independent \citep{oxley11}.
\end{definition}

We also call it the column matroid of $B$.
Since its ground set consists of column indices, repeated columns of $B$ give distinct elements of the ground set.

\begin{definition}\label{def:parity}
Let $B$ be a matrix over $\mathbb{F}_q$ whose column indices are partitioned into pairs.
The \emph{linear matroid parity problem}, also called the \emph{linear matroid matching problem},
asks for the maximum number $\nu$ of pairs
such that the $2\nu$ columns indexed by them are linearly independent.
\end{definition}

Linear matroid parity generalizes maximum matching in graphs.
\citet{lovasz80} proved a min--max formula for $\nu$,
and \citet{lovasz81} gave a polynomial-time algorithm to solve the problem.
See also \citep{lovasz86} and \citep[Chapter 43]{schrijver03}.
For a matrix with $m$ columns and rank $r$, the augmenting path algorithm of
\citet{gabowstallmann86} uses $O(mr^3)$ field operations,
and $O(mr^\omega)$ with fast matrix multiplication,
where $\omega$ is the exponent of matrix multiplication, as recorded in \citep[Section 1]{hamada26}
(both bounds assume the uniform cost model).
The shortest augmenting path algorithm of \citet{hamada26}, given in a recent preprint,
uses $O(mr^2\log r)$ field operations,
and $O(mr^2)$ with fast matrix multiplication.
Randomized algebraic algorithms that run in subcubic time are given in \citep{cheung14}.
They are analyzed over a sufficiently large field.
Hence, over a small field $\mathbb{F}_q$ they work in an extension field,
and their operation counts refer to that field.
In our application (Remark \ref{rem:AI1}), both $m$ and $r$ are at most twice
the size of the repair group, and inserting these values in the above bounds of \citet{hamada26}
gives the upper bounds $O(|K|^3\log|K|)$ and $O(|K|^3)$, respectively, where $|K|$ is that size.

\begin{definition}\label{def:intersection}
Let $B_1$ and $B_2$ be two matrices over $\mathbb{F}_q$
whose columns are indexed by the same finite set $R$.
The \emph{linear matroid intersection problem} asks for a largest set $Y\subseteq R$
such that the columns of $B_1[Y]$ are linearly independent and so are those of $B_2[Y]$.
\end{definition}

The linear matroid intersection problem is the special case of the linear matroid parity problem
for the block diagonal matrix $\operatorname{diag}(B_1,B_2)$,
in which the pairs are formed by the $\ell$th columns of the two blocks, $\ell\in R$.
Its optimal value is given by the following min--max formula.

\begin{proposition}[Matroid intersection theorem {\citep{edmonds70}, \citep[Chapter 41]{schrijver03}}]\label{prop:edmonds}
Let $B_1$, $B_2$ and $R$ be as in Definition \ref{def:intersection}.
Then the maximum size of a set $Y\subseteq R$ such that both $B_1[Y]$ and $B_2[Y]$
have linearly independent columns is
\begin{equation*}
  \min_{T\subseteq R}\Bigl\{\operatorname{rank}B_1[T]+\operatorname{rank}B_2[R\setminus T]\Bigr\}.
\end{equation*}
\end{proposition}

For two linear matroids on the same ground set of $m$ elements and of rank at most $r$,
\citet{cunningham86} gives an algorithm that uses $O(mr^2\log r)$ field operations,
and \citet{gabow96} achieve $O(mr^{1.62})$ field operations with fast matrix multiplication
(Coppersmith--Winograd exponent).
These algorithms only use linear algebra over the field of the matrices.
Faster randomized algebraic algorithms for linear matroid intersection exist \citep{harvey09},
but over small fields they require extension fields, and we do not use them.

\section{Computing a reduced repair group and a minimal set of observables}\label{sec3}
In this section we present a procedure that,
given a repair group $K$ for a set $I$ of erasures,
computes a reduced repair group $J$, with $I\subseteq J\subseteq K$,
together with a set of vectors of $C^{\perp s}$
that defines a smallest set of observables for correcting the erasures at $I$
in the stabilizer measurement model (Subsection \ref{sec213}).

The procedure consists of three steps:
(1) a decomposition of the space of stabilizer vectors supported on $K$;
(2) the choice of an admissible set of qudits, namely the qudits to be removed;
and (3) an elimination step based on Gaussian elimination.
The reduced repair group is introduced in Definition \ref{def:reduced},
and the properties of the procedure are collected in Theorem \ref{thm1}.

Throughout this section, $C \subseteq \mathbb{F}_q^{2n}$ is an $\mathbb{F}_q$-linear space
with $C \supsetneq C \cap C^{\perp s}$,
$c$ is as given in (\ref{def:c}), $C'$, $\alpha$ and $\beta$ are as in Subsection \ref{sec212},
and $Q(C)$ is the EAQECC defined by $C$ (Definition \ref{def:QC}).
We fix two sets $\emptyset\neq I\subsetneq K\subseteq\{1,\ldots,n\}$
such that $K$ is a repair group for $I$ (Definition \ref{def:repair}).
We do not consider the case $K=I$, in which there is nothing to reduce.
However, the reduced repair group $J$ may be equal to $I$.
Following the convention of Subsection \ref{sec211},
in this section every shortened code $\sigma_T(\cdot)$ of a subspace of $\mathbb{F}_q^{2n}$
is regarded as a subspace of $\mathbb{F}_q^{2n}$ by padding its vectors with zeros.
Punctured codes keep their own ambient space.
For vectors $\mathbf{v}_1,\ldots,\mathbf{v}_m$ over $\mathbb{F}_q$,
$\langle\mathbf{v}_1,\ldots,\mathbf{v}_m\rangle$ denotes the $\mathbb{F}_q$-linear space spanned by them.
For a matrix $B$, $\operatorname{rows}(B)$ denotes the set of its rows
and $\langle B \rangle$ its row space, that is, the space spanned by its rows.
The columns of a matrix with $2n$ columns are indexed by $1,\ldots,2n$,
so that the qudit $\ell\in\{1,\ldots,n\}$ corresponds to the columns $\ell$ and $n+\ell$,
which carry its two symplectic coordinates.

\emph{Step (1). Decomposition.}
Define
\begin{equation*}
  W = \sigma_{K \setminus I}(C^{\perp s}).
\end{equation*}
The spaces $\sigma_I(C^{\perp s})$ and $W$ are subspaces of $\sigma_K(C^{\perp s})$,
and their intersection is zero, because their nonzero vectors are supported
on the disjoint sets $I$ and $K\setminus I$, respectively.
Hence their sum is direct, and we can choose a complement $V$ of
$\sigma_I(C^{\perp s}) \oplus W$ in $\sigma_K(C^{\perp s})$, so that
\begin{equation}\label{eq:decomposition}
  \sigma_K(C^{\perp s}) = \sigma_I(C^{\perp s}) \oplus W \oplus V.
\end{equation}
The complement $V$ is not unique, and all the results below hold for every choice of $V$.
Let $B_I$, $B_W$ and $B_V$ be matrices whose rows form bases of
$\sigma_I(C^{\perp s})$, $W$ and $V$, respectively.
Assuming that a basis of $\sigma_K(C^{\perp s})$
has been precomputed for each repair group $K$,
the matrices $B_I$, $B_W$ and $B_V$ can be computed from it with $O(|K|^3)$
operations.

\emph{Step (2). Admissible sets.}
In order to remove a qudit $\ell \in K \setminus I$ from the measurement process,
both of its symplectic coordinates must vanish in all the vectors that define the observables,
that is, the columns $\ell$ and $n+\ell$ must be eliminated simultaneously.
For $P \subseteq K \setminus I$, let
\begin{equation*}
  \Gamma_P=\{\ell,\ n+\ell : \ell\in P\}\subseteq\{1,\ldots,2n\}
\end{equation*}
be the set of the $2|P|$ column indices corresponding to the qudits in $P$.
The following lemma characterizes when these columns can be removed from $B_V$
by means of the rows of $B_W$.

\begin{lemma}\label{lem:cancel}
Let $P \subseteq K \setminus I$.
There exists a matrix $M \in \mathbb{F}_q^{\dim V \times \dim W}$ such that
all the columns of $B_V - M B_W$ indexed by $\Gamma_P$ are zero
if and only if $\pi_P(V)\subseteq\pi_P(W)$.
\end{lemma}
\begin{proof}
The row spaces of $B_V[\Gamma_P]$ and $B_W[\Gamma_P]$ are $\pi_P(V)$ and $\pi_P(W)$,
up to the ordering of the coordinates.
A matrix $M$ as in the statement exists if and only if $B_V[\Gamma_P]=M B_W[\Gamma_P]$,
that is, if and only if every row of $B_V[\Gamma_P]$ is a linear combination of the rows of
$B_W[\Gamma_P]$, which is equivalent to $\pi_P(V)\subseteq\pi_P(W)$.
\end{proof}

The condition of Lemma \ref{lem:cancel} is in fact exactly the condition for $K\setminus P$ to be a repair group.

\begin{proposition}\label{prop:exact}
Let $P \subseteq K \setminus I$. The following statements are equivalent:
\begin{enumerate}
\item $K\setminus P$ is a repair group for $I$.
\item $\pi_P(\sigma_K(C^{\perp s}))=\pi_P(W)$.
\item $\pi_P(V)\subseteq\pi_P(W)$.
\end{enumerate}
In particular, the condition of Lemma \ref{lem:cancel} does not depend on the choice of $V$.
\end{proposition}
\begin{proof}
Set $U=\sigma_K(C^{\perp s})$ and $U_0=\sigma_{K\setminus P}(C^{\perp s})\subseteq U$.
Since $K$ is a repair group for $I$, Theorem \ref{thm:EAlocal} and Proposition \ref{prop:KI}
show that $K\setminus P$ is a repair group for $I$ if and only if $\pi_I(U_0)=\pi_I(U)$.
The kernel of the restriction of $\pi_I$ to $U$ is $W$, and therefore this holds if and only if $U=U_0+W$.
Since the kernel of the restriction of $\pi_P$ to $U$ is $U_0$, we have $U=U_0+W$ if and only if $\pi_P(U)=\pi_P(W)$.
Finally, $\pi_P(\sigma_I(C^{\perp s}))=\{\mathbf{0}\}$ and (\ref{eq:decomposition}) give
$\pi_P(U)=\pi_P(W)+\pi_P(V)$, which proves the equivalence of the last two statements.
\end{proof}

\begin{definition}\label{def:admissible}
A set $P \subseteq K \setminus I$ is called \emph{admissible}
if $\dim \pi_P(W) = 2|P|$, that is, if $\pi_P(W)=\mathbb{F}_q^{2|P|}$, or,
equivalently, if the $2|P|$ columns of $B_W$ indexed by $\Gamma_P$ are linearly independent.
\end{definition}

The empty set is admissible.
If $P$ is admissible, then $\pi_P(V)\subseteq\mathbb{F}_q^{2|P|}=\pi_P(W)$,
and Lemma \ref{lem:cancel} shows that the columns indexed by $\Gamma_P$ can be removed from $B_V$.
Admissibility is sufficient, but not necessary, for this cancellation.
It does not depend on $V$, and we adopt it because it can be described by means of a linear matroid.
Indeed, denote by $\mathbf{c}_j$ the $j$th column of $B_W$, $1\le j\le 2n$.
A set $P$ is admissible if and only if the $2|P|$ columns in the pairs
$\{\mathbf{c}_\ell, \mathbf{c}_{n+\ell}\}$, $\ell\in P$, are linearly independent.
Hence admissibility is feasibility for
the linear matroid parity problem (Definition \ref{def:parity})
on the column matroid of $B_W[\Gamma_{K\setminus I}]$ (Definition \ref{def:matroid}),
with the column indices paired by qudit, that is, with the pairs
$\{\ell, n+\ell\}$, $\ell \in K \setminus I$,
and a largest admissible set can be computed in deterministic polynomial time
(Subsection \ref{sec2m} and Remark \ref{rem:AI1}).

Note that the ground set of this matroid consists of column indices, and that
the admissible sets of qudits need not be the independent sets of a matroid on $K\setminus I$.
For instance, for $q=2$, $C=C_X\times C_Z$ with $C_X=\langle(1,0,1,0)\rangle$ and $C_Z=\langle(1,1,0,0)\rangle$,
$I=\{4\}$ and $K=\{1,2,3,4\}$, the sets $\{1\}$ and $\{2,3\}$ are admissible,
but neither $\{1,2\}$ nor $\{1,3\}$ is admissible.
The construction below and Theorem \ref{thm1} are valid for every admissible set $P$,
and a larger $P$ gives a smaller upper bound on the size of the reduced repair group
(Theorem \ref{thm1}).
The closed-form bounds of Proposition \ref{prop:general} and of Section \ref{sec4}
are stated for the particular admissible sets specified there.

\emph{Step (3). Elimination.}
Let $P$ be an admissible set.
Since $\pi_P(W) = \mathbb{F}_q^{2|P|}$, the submatrix $B_W[\Gamma_P]$ has full column rank.
By Lemma \ref{lem:cancel}, we can eliminate the coordinates indexed by $\Gamma_P$ from $B_V$
by subtracting a linear combination of vectors of $W$. Namely, we set
\begin{equation}\label{eq:elimination_step}
\widetilde{B}_V = B_V - M B_{W},
\end{equation}
where $M \in \mathbb{F}_q^{\dim V \times \dim W}$ is a matrix such that
all the columns of $\widetilde{B}_V$ indexed by $\Gamma_P$ are zero.
Such a matrix $M$ is obtained by Gaussian elimination, and it need not be unique.
If $M_0$ is one solution, then the solutions are the matrices $M_0+L$
such that the rows of $L$ lie in the left kernel of $B_W[\Gamma_P]$.
Since this left kernel has dimension $\dim W-2|P|$,
the matrix $M$ is unique if and only if $\dim V=0$ or $\dim W=2|P|$.
Different solutions can give different reduced repair groups,
even for fixed $V$, $P$ and bases.
All the results below hold for every solution $M$.
Set $\widetilde{V}=\langle \widetilde{B}_V \rangle$.

\begin{lemma}\label{lem:tildeV}
With the above notation, the following statements hold.
\begin{enumerate}
\item\label{tv1} $\sigma_K(C^{\perp s}) = \sigma_I(C^{\perp s}) \oplus W \oplus \widetilde{V}$,
and $\dim\widetilde{V}=\dim V$.
\item\label{tv2} $\pi_I(\widetilde{V})=\pi_I(V)$.
\item\label{tv3} $\operatorname{supp}(\widetilde{V})\subseteq K\setminus P$.
\end{enumerate}
\end{lemma}
\begin{proof}
First, we prove Item (\ref{tv1}).
Every row of $\widetilde{B}_V$ differs from the corresponding row of $B_V$ by a vector of $W$.
Hence $\sigma_I(C^{\perp s}) + W + \widetilde{V}=\sigma_I(C^{\perp s}) + W + V=\sigma_K(C^{\perp s})$
by (\ref{eq:decomposition}).
Since $\widetilde{B}_V$ has $\dim V$ rows, we get
\begin{align*}
\dim\sigma_K(C^{\perp s})&\le \dim\sigma_I(C^{\perp s})+\dim W+\dim \widetilde{V}\\
&\le\dim\sigma_I(C^{\perp s})+\dim W+\dim V=\dim\sigma_K(C^{\perp s}).
\end{align*}
Thus both inequalities are equalities, which proves Item (\ref{tv1}).

Let us show Item (\ref{tv2}).
Every vector of $W$ is supported on $K\setminus I$, and thus $\pi_I(W)=\{\mathbf{0}\}$.
Hence $\pi_I(\mathbf{u}-\mathbf{w})=\pi_I(\mathbf{u})$ for every $\mathbf{u}\in V$ and $\mathbf{w}\in W$.

Finally, we prove Item (\ref{tv3}).
The rows of $\widetilde{B}_V$ belong to $\sigma_K(C^{\perp s})$, which implies that they are supported on $K$.
Moreover, their coordinates at the positions $\ell$ and $n+\ell$ are zero for every $\ell\in P$.
Therefore they are supported on $K\setminus P$, and so is $\widetilde{V}$.
\end{proof}

We can now define the main object of this section.

\begin{definition}\label{def:reduced}
The \emph{reduced repair group} associated with $V$, $P$ and $M$ is
\begin{equation}
J = I \cup \bigcup_{\mathbf{u} \in \operatorname{rows}(\widetilde{B}_V)} \operatorname{supp}(\mathbf{u})
= I\cup\operatorname{supp}(\widetilde{V}). \label{eq:j_definition}
\end{equation}
\end{definition}

Since $P\cap I=\emptyset$, Item (\ref{tv3}) of Lemma \ref{lem:tildeV} gives
\begin{equation*}
  I \subseteq J \subseteq K\setminus P\subseteq K.
\end{equation*}
Theorem \ref{thm1} below shows that the erasures at $I$ can be corrected
by measuring a smallest set of observables
which, among the $n$ codeword qudits, act only on the qudits in $J$.
The set $J$ depends on the choice of $V$, $P$ and $M$.
Theorem \ref{thm1} holds for every such choice, and in particular
$|J|-|I|\le|K\setminus I|-|P|$ for every admissible set $P$.
In contrast, the closed-form bounds on $|J|-|I|$ of Proposition \ref{prop:general} and of Section \ref{sec4}
(Theorem \ref{thm:iotaE} and Proposition \ref{prop:herm}) are stated only for the admissible sets $P$ specified there,
and then they hold for every $V$ and $M$.
They can fail for a smaller admissible set.
For instance, in Example \ref{ex:4222},
the choice $P=\emptyset$ returns $J=K$ for suitable $V$ and $M$,
and then $|J|-|I|=3$ exceeds the bound $\iota_e(K\setminus I)=2$ of Theorem \ref{thm:iotaE}.
The choice of $V$ is nevertheless immaterial for the optimum.
Indeed, by Proposition \ref{prop:exact}, if $J_0$ is a smallest repair group for $I$ with $I\subseteq J_0\subseteq K$,
then for every $V$ there is a matrix $M$ (with $P=\emptyset$) such that $J=J_0$.
Hence minimizing $|J|$ over $V$, $P$ and $M$ amounts to computing a smallest repair group inside $K$,
which is NP-hard in general (Remark \ref{rem:AI1}).
The admissible sets provide a polynomial-time guarantee, not the optimum.

\begin{theorem}\label{thm1}
  Keep the above notation and consider the maps $\alpha$ and $\beta$ defined in Subsection \ref{sec212}.
  Let $K$ be a repair group for a set $I$ of erasures,
  and let $J$ be the reduced repair group associated with $V$, $P$ and $M$.
  Then the following statements hold.
  \begin{enumerate}
    \item\label{l1} Once an admissible set $P$ is given, and assuming that a basis of $\sigma_K(C^{\perp s})$
      has been precomputed for each repair group $K$, the procedure above computes $J$
      and a basis of $\sigma_I(C^{\perp s}) \oplus \widetilde{V}$ with $O(|K|^3)$ operations.
      If, in addition, the values of $\alpha$ on that basis of $\sigma_K(C^{\perp s})$ are given,
      the set $O$ of Item (\ref{l2}) is obtained with $O(|K|^2c)$ further operations.
    \item\label{l4} $|J|\le|K| - |P|$.
    \item\label{l2} The set
      $O = \bigl\{ \beta(\mathbf{u}) : \mathbf{u}\in\operatorname{rows}(B_I)\cup\operatorname{rows}(\widetilde{B}_V) \bigr\}$
      defines a set of observables that suffices to correct the erasures at $I$,
      and $J$ is the set of codeword qudits measured by the decoder.
    \item\label{l3} $|O| = \dim \pi_I(C^{\perp s})$, which is the smallest possible number of observables
      for correcting the erasures at $I$ with the EAQECC $Q(C)$ (and any repair group)
      in the stabilizer measurement model.
    \item\label{l5} $J$ is a repair group for $I$, and the same holds for every $J'$ with $J\subseteq J'\subseteq K$.
      In particular, $K\setminus P$ is a repair group for $I$.
  \end{enumerate}
\end{theorem}
\begin{proof}
  Items (\ref{l1}) and (\ref{l4})
  follow from the construction above. In particular, $|J|\le|K|-|P|$ because $J\subseteq K\setminus P$.
  For the last claim of Item (\ref{l1}), note that, keeping track of the row operations,
  each row $\mathbf{u}$ of $B_I$ and $\widetilde{B}_V$ is a known linear combination
  of the at most $2|K|$ vectors of the precomputed basis of $\sigma_K(C^{\perp s})$.
  Since $\alpha$ is linear, $\alpha(\mathbf{u})\in\mathbb{F}_q^{2c}$ is the same linear combination of their images.
  As there are $\dim\pi_I(C^{\perp s})\le 2|K|$ such rows, this costs $O(|K|^2c)$ operations.

  To prove Item (\ref{l2}), recall that,
  as proved in \citep{qlrsurface25eprint} (see Theorem \ref{thm:minobs} and (\ref{eq:smallest})),
  a set of observables spanning a space $D$ with
  $D\oplus\sigma_{\{1,\ldots,n+c\}\setminus I}((C')^{\perp s})=(C')^{\perp s}$ suffices
  to correct the erasures at $I$, and the observables
  that act as the identity on all the erased qudits need not be measured.
Let $D=\beta\bigl(\sigma_I(C^{\perp s})\oplus\widetilde{V}\bigr)\subseteq(C')^{\perp s}$,
which is spanned by the observables of $O$.
By (\ref{eq1b}), $\pi_I(\sigma_K(C^{\perp s}))\subseteq\pi_I(C^{\perp s})\subseteq\pi_I(C+C^{\perp s})
=\pi_I(\sigma_K(C^{\perp s}))$, hence $\pi_I(\sigma_K(C^{\perp s}))=\pi_I(C^{\perp s})$.
The kernel of $\pi_I$ on $\sigma_K(C^{\perp s})$ is $W$. Thus, by Item (\ref{tv1}) of Lemma \ref{lem:tildeV},
$\pi_I$ maps $\sigma_I(C^{\perp s})\oplus\widetilde{V}$ bijectively onto $\pi_I(C^{\perp s})$.
Since $\pi_I\circ\beta=\pi_I$ on $C^{\perp s}$, $\pi_I$ maps $D$ bijectively onto $\pi_I((C')^{\perp s})$,
and the kernel of $\pi_I$ on $(C')^{\perp s}$ is $\sigma_{\{1,\ldots,n+c\}\setminus I}((C')^{\perp s})$.
Hence (\ref{eq:smallest}) holds with $D'=D$, and Theorem \ref{thm:minobs}, applied to $Q(C')$,
shows that $O$ suffices to correct the erasures at $I$.
The measured codeword qudits are $\operatorname{supp}(\sigma_I(C^{\perp s}))\cup\operatorname{supp}(\widetilde{V})\subseteq J$.
Conversely, let $i\in I$. We have $\pi_{\{i\}}(C+C^{\perp s})\neq\{\mathbf{0}\}$.
Indeed, otherwise $\pi_{\{i\}}(C)=\{\mathbf{0}\}$, and (\ref{eq:dual}) gives
$\sigma_{\{i\}}(C^{\perp s})=\pi_{\{i\}}(C)^{\perp s}=\mathbb{F}_q^2\subseteq\pi_{\{i\}}(C+C^{\perp s})$, a contradiction.
By (\ref{eq1b}) some vector of $\sigma_K(C^{\perp s})=\sigma_I(C^{\perp s})\oplus W\oplus\widetilde{V}$ is nonzero at $i$.
Since $W$ vanishes on $I$, the position $i$ lies in $\operatorname{supp}(\sigma_I(C^{\perp s}))\cup\operatorname{supp}(\widetilde{V})$.
Thus the set of measured codeword qudits is exactly $J$.

Let us show Item (\ref{l3}). Since the kernel of the restriction of $\pi_I$ to $\sigma_K(C^{\perp s})$ is $\sigma_{K \setminus I}(C^{\perp s})$,
  we have
  \begin{equation*}
    \dim \sigma_K(C^{\perp s}) = \dim \sigma_{K \setminus I}(C^{\perp s}) +
    \dim \pi_I (\sigma_K(C^{\perp s})).
  \end{equation*}
  Moreover, by (\ref{eq1b}),
  \begin{equation*}
    \pi_I(\sigma_K(C^{\perp s})) \subseteq \pi_I(C^{\perp s}) \subseteq \pi_I(C + C^{\perp s}) = \pi_I(\sigma_K(C^{\perp s})),
  \end{equation*}
  which implies $\pi_I(\sigma_K(C^{\perp s})) = \pi_I(C^{\perp s})$.
  Therefore, by Item (\ref{tv1}) of Lemma \ref{lem:tildeV} and (\ref{eq:decomposition}),
  \begin{multline*}
    |O| = \dim \sigma_I(C^{\perp s}) + \dim \widetilde{V} = \dim \sigma_I(C^{\perp s}) + \dim V\\
    =\dim \sigma_K(C^{\perp s}) - \dim \sigma_{K \setminus I}(C^{\perp s}) = \dim \pi_I (C^{\perp s}),
  \end{multline*}
  which is the smallest possible size by (\ref{eq:smallest2}).

  Finally, we prove Item (\ref{l5}).
  Since $\pi_I$ vanishes on $W$, Item (\ref{tv1}) of Lemma \ref{lem:tildeV} gives
  $\pi_I(\sigma_I(C^{\perp s})\oplus\widetilde{V})=\pi_I(\sigma_K(C^{\perp s}))=\pi_I(C+C^{\perp s})$ by (\ref{eq1b}).
  As
  \begin{equation*}
  \sigma_I(C^{\perp s})\oplus\widetilde{V}\subseteq\sigma_J(C^{\perp s})\subseteq\sigma_{J'}(C^{\perp s})\subseteq\sigma_K(C^{\perp s}),
  \end{equation*}
  we obtain $\pi_I(\sigma_{J'}(C^{\perp s}))=\pi_I(C+C^{\perp s})$,
  that is, (\ref{eq1b}) holds for $J'$.
  By Theorem \ref{thm:EAlocal} and Proposition \ref{prop:KI} (which covers $J'=I$), $J'$ is a repair group for $I$.
  The last claim follows from $J\subseteq K\setminus P$.
\end{proof}

The last result in this section is the following one.

\begin{proposition}\label{prop:general}
Let $R=K\setminus I$, $m=|R|$, $k=\dim\pi_R(C)$,
and let $\nu$ be the maximum size of an admissible set.
Then $\dim W = 2m-k$ and
\begin{equation}\label{eq:nubounds}
  \max\{0,m-k\}\le\nu\le m-\lceil k/2\rceil .
\end{equation}
Moreover, an admissible set $P_b$ with $|P_b|\ge m-k$ is computed by Gaussian elimination
with $O(|K|^3)$ operations.
If $P$ is $P_b$ or a largest admissible set, the set $J$ of Theorem \ref{thm1} satisfies
\begin{equation}\label{eq:boundgen}
  |J|-|I|\le |R|-|P|\le \min\bigl\{|K\setminus I|,\ \dim\pi_{K\setminus I}(C)\bigr\}.
\end{equation}
\end{proposition}
\begin{proof}
By (\ref{eq:dual}) considered in $\mathbb{F}_q^{2|R|}$, $W=\sigma_R(C^{\perp s})=\pi_R(C)^{\perp s}$,
which gives $\dim W=2m-k$.
The $2\nu$ columns of an admissible set of size $\nu$ are independent, hence $2\nu\le\dim W$,
which is the upper bound in (\ref{eq:nubounds}).
Let $\mathcal{B}$ be a set of $\dim W$ independent columns of $B_W$ (a column basis),
let $b_2$ be the number of $\ell\in R$ with both columns in $\mathcal{B}$
and $b_1$ the number with exactly one.
Then $2b_2+b_1=\dim W$ and $b_1+b_2\le m$, thus $b_2\ge\dim W-m=m-k$.
These $b_2$ qudits form an admissible set $P_b$.
A column basis is found by Gaussian elimination in $O(|K|^3)$ operations.
Finally, (\ref{eq:boundgen}) follows from Item (\ref{l4}) of Theorem \ref{thm1}.
\end{proof}

\begin{remark}\label{rem:unassisted}
  For the entanglement-unassisted case $C^{\perp s} \subsetneq C \subsetneq \mathbb{F}_q^{2n}$, we have $c=0$ because $C \cap
  C^{\perp s} = C^{\perp s}$.
  Consequently, $C' = C$, $(C')^{\perp s} = C^{\perp s}$, and $\beta$ is the identity map of $C^{\perp s}$.
  Thus, the set of observables in Item (\ref{l2}) of Theorem \ref{thm1} becomes
  \begin{equation*}
    O = \operatorname{rows}(B_I) \cup \operatorname{rows}(\widetilde{B}_V).
  \end{equation*}
  Here $Q(C)$ is an entanglement-unassisted quantum stabilizer code, and the minimality of $|O|$ in Item (\ref{l3}) is justified
  directly by \citep[Theorem 5]{qlrsurface25eprint}.
  The algorithmic procedure in Section \ref{sec3} and the bound $|J| \leq |K| - |P|$ in Theorem \ref{thm1} remain valid without
  change.
\end{remark}

In the next remark we discuss the computational cost of finding a largest admissible set,
and we show that, even with such a set, the procedure need not return a smallest repair group,
whose computation is NP-hard.

\begin{remark}\label{rem:AI1}
  The size $|P|$ controls the bound $|J|\le|K|-|P|$ in Theorem \ref{thm1}.
  Finding
  \begin{equation}
    \max_{P' \subseteq (K \setminus I)} \bigl\{ |P'| : \dim \pi_{P'}(\sigma_{K \setminus I}(C^{\perp s})) = 2|P'| \bigr\}
    \label{eq:maxp}
  \end{equation}
  is an instance of the linear matroid parity problem (linear matroid matching) on
  the pairs of column vectors $\bigl\{\{\mathbf{c}_\ell, \mathbf{c}_{n+\ell}\} : \ell \in K \setminus I\bigr\}$.
  Here the ground set has $m=2|K\setminus I|\le2|K|$ elements and the rank is $r=\dim W\le2|K|$.
  Hence, by the bounds recalled in Subsection \ref{sec2m},
  the maximum (\ref{eq:maxp}) is attained with $O(|K|^4)$ operations by the algorithm of \citet{gabowstallmann86},
  and with $O(|K|^3\log |K|)$ operations, or $O(|K|^3)$ with fast matrix multiplication,
  by the shortest augmenting path algorithm of \citet{hamada26}.

  Even when $P$ attains the maximum (\ref{eq:maxp}),
  the size of the computed set $J$ can be strictly larger than the minimum possible size
  of a repair group for correcting the erasures at $I$.
For example, let $q=5$ and $n=4$,
  and let $C=C_X\times C_X$, where $C_X$ has parity-check matrix
  $\left(\begin{smallmatrix}1&1&0&0\\0&1&1&1\end{smallmatrix}\right)$.
  Then $C_X\cap C_X^{\perp e}=\langle(1,4,3,3)\rangle$, where ${}^{\perp e}$ denotes the Euclidean dual, and hence $c=1$.
  For $I=\{1\}$ and $K=\{1,2,3,4\}$, the set $P=\{2\}$ attains the maximum $|P|=1$
  and gives $J=\{1,3,4\}$ with $|J|=3$,
  whereas $\{1,2\}$ is also a repair group for $I$ and has size $2$.
  The three singletons of $\{2,3,4\}$ are all maximum admissible sets, and
  $\{3\}$ and $\{4\}$ give $J=\{1,2\}$.

  Computing a smallest repair group for given $C$, $I$ and $K$
  is at least as hard as finding a nonzero vector of minimum Hamming weight
  in a linear code, and this is NP-hard \citep{vardy97, austrin14}. Hence computing a smallest repair group is NP-hard over every fixed finite field,
  under polynomial-time Turing reductions, even for EAQECCs with $c>0$.
  In the classical case, the smallest recovery sets of a coordinate $i$ are obtained from the codewords of the dual code
  of minimum Hamming weight among those whose $i$th coordinate is nonzero \citep[Proposition 1 and Corollary 2]{sharp},
  and the relation with the computation of the minimum distance is pointed out in \citep[Remark 4(2)]{sharp}.
  The reduction below embeds this classical problem into CSS codes defined by two equal codes.

  To justify the above claim, let $C_X \subsetneq \mathbb{F}_q^n$ be a linear code
  in which a nonzero vector of minimum Hamming weight is sought.
  Without loss of generality, $C_X$ has no all-zero coordinate positions (which can be verified in $O(n^2)$ time).
  Let
  $$\widehat{C}_X=\bigl\{(\mathbf{x},b,0):\mathbf{x}\in C_X,\ b\in\mathbb{F}_q\bigr\}\subseteq\mathbb{F}_q^{n+2}$$
  and $C=\widehat{C}_X^{\perp e}\times\widehat{C}_X^{\perp e}$, so that $C^{\perp s}=\widehat{C}_X\times\widehat{C}_X$.
  The unit vector $\mathbf{e}_{n+2}$ lies in $\widehat{C}_X^{\perp e}$ but not in $\widehat{C}_X$,
  and $\mathbf{e}_{n+1}$ lies in $\widehat{C}_X$ but not in $\widehat{C}_X^{\perp e}$, because $\mathbf{e}_{n+1}\cdot\mathbf{e}_{n+1}=1$.
  Hence $C\not\subseteq C^{\perp s}$, which is our standing assumption, and $C^{\perp s}\not\subseteq C$, that is, $c>0$.
  Let $K=\{1,\ldots,n+2\}$ and $i\in\{1,\ldots,n\}$.
  Since $\pi_{\{i\}}(\widehat{C}_X)=\mathbb{F}_q$, we have $\pi_{\{i\}}(C+C^{\perp s})=\mathbb{F}_q^2$.
  Hence, by Theorem \ref{thm:EAlocal} (in the form (\ref{eq1b})) and Proposition \ref{prop:KI},
  a set $J$ with $i\in J\subseteq K$ is a repair group for $\{i\}$ if and only if $\pi_{\{i\}}(\sigma_J(\widehat{C}_X))=\mathbb{F}_q$,
  that is, if and only if $J$ contains the support of a vector of $\widehat{C}_X$ whose $i$th coordinate is nonzero.
  In particular, $K$ is a repair group for $\{i\}$,
  and a smallest repair group $K_i$ for $\{i\}$ is the support of a vector of $\widehat{C}_X$
  of minimum Hamming weight among those with nonzero $i$th coordinate.
  Such a vector has $b=0$, since otherwise replacing $b$ by $0$ gives a vector of smaller weight
  with the same $i$th coordinate. Hence $K_i\subseteq\{1,\ldots,n\}$.

  Now, for each $i\in\{1,\ldots,n\}$, we find a smallest repair group $K_i$ for the EAQECC
  $Q(C)$ and the set $I=\{i\}$ of erasures, and then a vector $\mathbf{x}_i \in C_X$
  with nonzero $i$th coordinate and support contained in $K_i$, and hence equal to $K_i$,
  by solving the resulting homogeneous system of
  $\dim C_X^{\perp e}$ linear equations in $|K_i|$ unknowns.
  Any vector of minimum Hamming weight among $\mathbf{x}_1, \ldots, \mathbf{x}_n$
  is a nonzero vector of minimum Hamming weight in $C_X$,
  because every nonzero vector of $C_X$ has a nonzero coordinate $i$, and then its weight is at least $|K_i|$.
  Therefore, if a deterministic algorithm finds a smallest
  repair group for arbitrary $C \subseteq \mathbb{F}_q^{2n}$,
  $I$ and $K$ (even if only codes with $c>0$ are allowed) with $O(f(n))$ complexity,
  then a nonzero vector of minimum Hamming weight in
  an arbitrary linear code
  of length $n$ can be computed with $O(n f(n+2) + n^4)$ complexity.

  The above reduction also preserves approximation ratios.
  Every repair group $J$ for $\{i\}$ with $i\in J\subseteq K$ contains the support of a vector $(\mathbf{x},b,0)\in\widehat{C}_X$
  whose $i$th coordinate is nonzero, and then $(\mathbf{x},0,0)\in\widehat{C}_X$.
  Such a vector $\mathbf{x}\in C_X$, whose weight is at most $|J|$, is found by linear algebra.
  Hence, if an algorithm returns, for each $i\in\{1,\ldots,n\}$, a repair group $J_i$ for $\{i\}$
  with $|J_i|\le\rho|K_i|$ for a constant $\rho\ge1$
  (this holds, in particular, if $|J_i|-1\le\rho(|K_i|-1)$),
  then it yields a nonzero vector of $C_X$ whose weight is at most $\rho$ times the minimum Hamming weight of $C_X$.
  Over every fixed finite field, approximating the minimum Hamming weight of a linear code
  within some constant factor greater than $1$ is NP-hard under deterministic polynomial-time reductions
  \citep[Theorem 1]{austrin14}.
  Since the minimum Hamming weight of the tensor product of two linear codes is the product of their minimum Hamming weights
  \citep[Fact 1]{austrin14}, the same holds for every constant factor, as noted in \citep[Section I]{austrin14}.
  Therefore, unless $\mathrm{P}=\mathrm{NP}$, no polynomial-time algorithm computes repair groups
  whose size, or whose number of additional codeword qudits, is always within a constant factor of the smallest one.
\end{remark}

In the next remark we explain how to apply our procedure to the auxiliary code $Q(C')$
in order to reduce jointly the measured codeword qudits and the measured qudits held by the decoder.

\begin{remark}\label{rem:joint}
The set $J$ only accounts for the $n$ codeword qudits.
In general, some of the $c$ qudits held by the decoder also lie in the supports of the vectors in the set of observables $O$,
and they are therefore measured.
One can reduce both kinds of qudits jointly as follows.
Let $C'\subseteq\mathbb{F}_q^{2(n+c)}$ be as in Section \ref{sec2}.
It defines an unassisted stabilizer code $Q(C')$, and Remark \ref{rem:unassisted} applies to it.
Set $K'=K\cup\{n+1,\ldots,n+c\}\supsetneq I$.
By (\ref{eq1b}) and Proposition \ref{prop:reduction}, $K'$ is a repair group for $I$ in $Q(C')$.
We apply Theorem \ref{thm1} to $(C',I,K')$.
It gives $J'\subseteq K'$ with $I\subseteq J'$, a set of observables of the minimum size
$\dim\pi_I(C^{\perp s})$, and $|J'|-|I|\le|K'|-|I|-|P'|$,
where $P'$ is an admissible set of $\sigma_{K'\setminus I}((C')^{\perp s})$.
Since $P'$ may now contain qudits of the decoder, $J'$ accounts for the measured qudits among all the $n+c$ qudits.
A largest $P'$ is computed in $O(|K'|^4)$ operations by \citep{gabowstallmann86}, and in
$O(|K'|^3\log|K'|)$ by \citep{hamada26}, with $|K'|=|K|+c$.
Since $C'$ need not be a product of two codes, the matroid intersection algorithms may not apply.
This reduction runs in polynomial time but does not necessarily attain a minimum, as in Remark \ref{rem:AI1}; moreover, the result depends on the choice of $C'$, that is, of $\alpha$.

For a fixed space $L\subseteq C^{\perp s}$ of measured vectors, the dependence on $C'$ is explicit.
More precisely, if $d_B=\dim L-\dim(L\cap C)$ and $2s$ is the rank of the symplectic form restricted to $L$,
then the minimum, over all choices of $C'$, of the number of qudits held by the decoder
in the support of $\alpha(L)$ is $d_B-s$.
Indeed, the graph of $\alpha$ is isotropic and $\ker\alpha=C\cap C^{\perp s}$,
and hence $\alpha$ induces an anti-isometry from $C^{\perp s}/(C\cap C^{\perp s})$ onto $\mathbb{F}_q^{2c}$,
and the admissible choices of $\alpha$ differ by symplectic transformations of $\mathbb{F}_q^{2c}$.
The image $\alpha(L)$ has dimension $d_B$ and a radical of dimension $d_B-2s$.
A subspace of this kind supported on $u$ qudits satisfies $d_B-2s\le u-s$,
and equality is attained by completing a symplectic basis.
\end{remark}

\section{\texorpdfstring{Dimension locality of QECCs from Euclidean and Hermitian orthogonality}{Dimension locality of QECCs from Euclidean and Hermitian orthogonality}}\label{sec4}

We begin with classical codes, which motivate this section.
Let $E\subseteq\mathbb{F}_q^n$ be a linear code,
and let $K\subseteq\{1,\ldots,n\}$ be a repair group for a set $I\subsetneq K$ of erasures,
that is, the symbols at $I$ of every codeword of $E$ can be recovered from its symbols at $K\setminus I$.
The symbols at $K\setminus I$ lying outside an information set of $\pi_{K\setminus I}(E)$ are not needed:
if $A$ is an information set of the punctured code $\pi_{K\setminus I}(E)$,
which is found by Gaussian elimination, then $J=I\cup A$ is also a repair group for $I$, and
\begin{equation*}
  |J|-|I|=\dim\pi_{K\setminus I}(E)\le\dim\pi_K(E).
\end{equation*}
This gives a repair group strictly smaller than $K$ if and only if $\dim\pi_{K\setminus I}(E)<|K\setminus I|$.
In the $(r,\delta)$-locality of \citep[Definition 1]{prakash12}, where $|K|\le r+\delta-1$ and
the minimum distance of $\pi_K(E)$ is at least $\delta$,
the number $|J|-|I|$ of additional codeword symbols is bounded by the looser quantity $|K|-\delta+1\le r$.

A code has in general many repair groups, and the one that is used depends on the erased positions.
For this reason, locality parameters refer to a collection $\mathcal{K}$ of repair groups fixed in advance,
and they take the worst case over $K\in\mathcal{K}$.
In this paper, the \emph{dimension locality} of a classical code $E$ with a collection $\mathcal{K}$ of repair groups is
$\max_{K\in\mathcal{K}}\dim\pi_K(E)$.
By the previous paragraph, it bounds the number $|J|-|I|$ of additional symbols
for every $K\in\mathcal{K}$ and every set $I$ of erasures repaired by $K$.

In \citep[Definition 1]{sasidharan15} and \citep[Definition 4]{grezet19}, a code has
all-symbol $(r,\delta)$-locality (respectively, all-symbol dimension-locality $(\kappa,\delta)$)
if every coordinate $i$ lies in a set $K$ (in \citep{sasidharan15}, in the support of $\pi_K(E)$)
such that the punctured code $\pi_K(E)$
has dimension at most $r$ (respectively, at most $\kappa$) and minimum distance at least $\delta$.
These parameters are upper bounds attached to the code,
whereas the above maximum depends on the chosen collection $\mathcal{K}$.
If $\mathcal{K}$ contains, for every coordinate $i$, a set $K$ such that $i$ lies in the support of $\pi_K(E)$
and the minimum distance of $\pi_K(E)$ is at least $\delta$,
then $\max_{K\in\mathcal{K}}\dim\pi_K(E)$ is a valid value of $r$ in \citep[Definition 1]{sasidharan15}
and of $\kappa$ in \citep[Definition 4]{grezet19}.
The name dimension locality is taken from \citep[Definition 4]{grezet19}.
This notion differs from the information locality of \citep[Definition 2]{kamath14},
where the local codes have length at most $r+\delta-1$ and minimum distance at least $\delta$,
and their supports cover an information set.
What \citep[Section 3.1]{ilocality26} calls classical information locality $(r,\delta)$
is the code-level notion of \citep[Definition 4]{grezet19} recalled above,
with the dimension bound $\dim\pi_K(E)\le r$, and not the notion of \citep{kamath14}.
The dimension locality of this paper is the value of $r$ supplied by a fixed collection $\mathcal{K}$.

The aim of this section is to define analogous notions for quantum codes.
The advantage of dimension locality over the conventional
symbol locality \citep[Definition 10]{qlrc24}
was demonstrated in the quantum setting by an explicit example in \citep{ilocality26}.
However, the notion of dimension locality for quantum local recovery was clarified there only for
Hermitian orthogonality without entanglement assistance.

For a quantum code, the role of $\dim\pi_K(E)$ is played by an upper bound on $|J|-|I|$,
where $J$ is the reduced repair group computed in Section \ref{sec3}.
By Item (\ref{l4}) of Theorem \ref{thm1},
\[
|J|-|I|\le|K\setminus I|-|P|,
\]
and therefore we need admissible sets $P$ as large as possible.
We study this problem for two families of EAQECCs, with or without entanglement assistance,
and, for each of them, we describe the largest admissible sets
and deduce closed-form upper bounds on $|J|-|I|$.

In Subsection \ref{sec41} we consider EAQECCs from Euclidean orthogonality,
that is, CSS codes $C=C_X\times C_Z$ defined by two linear codes $C_X,C_Z\subseteq\mathbb{F}_q^n$.
For these codes, the two symplectic coordinates of a qudit lie in two different codes,
and a largest admissible set is obtained by solving a linear matroid intersection problem.
This will lead to the bound $|J|-|I|\le\iota_e(K\setminus I)$ of Theorem \ref{thm:iotaE} below
(where $\iota_e$ will be introduced in Definition \ref{def:iotaE}),
to a simpler closed-form bound in terms of the sum code $C_X+C_Z$,
and to the dimension locality and the refined dimension locality
introduced later in Definitions \ref{def:ilE} and \ref{def:ilE2}.

In Subsection \ref{sec42} we consider EAQECCs from Hermitian orthogonality,
defined by an $\mathbb{F}_{q^2}$-linear code $D\subseteq\mathbb{F}_{q^2}^n$.
For these codes only one matroid is involved,
a largest admissible set is found by Gaussian elimination,
and the resulting bound $|J|-|I|\le\dim_{\mathbb{F}_{q^2}}\pi_{K\setminus I}(D)$,
which will be proved in Proposition \ref{prop:herm},
will lead to the dimension locality introduced later in Definition \ref{def:ilH}.

\subsection{Euclidean orthogonality: CSS codes defined by two linear codes}\label{sec41}
Let $C_X$ and $C_Z$ be $\mathbb{F}_q$-linear subspaces
of $\mathbb{F}_q^n$ and set $C = C_X \times C_Z$.
We have
$C^{\perp s} = C_Z^{\perp e} \times C_X^{\perp e}$,
where, as before, ${}^{\perp e}$ denotes the Euclidean dual.
We also assume $C \supsetneq C^{\perp s} \cap C$.
The EAQECC $Q(C_X \times C_Z)$ has parameters
$[[n, c+\dim C_X + \dim C_Z - n, d; c]]_q$,
where $d$ is the minimum Hamming weight of
$(C_X \setminus C_Z^{\perp e}) \cup (C_Z \setminus C_X^{\perp e})$
and $$c = \dim C_Z^{\perp e} - \dim (C_Z^{\perp e} \cap C_X) = \dim C_X^{\perp e} - \dim (C_X^{\perp e} \cap
C_Z).$$

We use the notation of Section \ref{sec3}, with $W=\sigma_{K\setminus I}(C^{\perp s})$,
and we look for an admissible set $P$ (Definition \ref{def:admissible}) of maximum size.
Set $R=K\setminus I$,
\begin{equation*}
  W_Z=\sigma_R(C_Z^{\perp e}),\qquad W_X=\sigma_R(C_X^{\perp e}),
\end{equation*}
so that $W=W_Z\times W_X$, because $C^{\perp s}=C_Z^{\perp e}\times C_X^{\perp e}$.

For $S\subseteq\{1,\ldots,n\}$, write $k_X(S)=\dim\pi_S(C_X)$ and $k_Z(S)=\dim\pi_S(C_Z)$.
Recall that an information set of the punctured code $\pi_S(C_X)\subseteq\mathbb{F}_q^{|S|}$,
whose coordinates are indexed by $S$, is a set $A\subseteq S$
with $|A|=k_X(S)$ such that $\pi_A$ is injective on $\pi_S(C_X)$.
In other words, an information set of $\pi_S(C_X)$ is a smallest set of positions of $S$
whose symbols determine every codeword of $\pi_S(C_X)$, and its size is $k_X(S)=\dim\pi_S(C_X)$.
Information sets of $\pi_S(C_Z)$ are defined in the same way, and their size is $k_Z(S)$.
Definition \ref{def:iotaE} below considers both codes simultaneously:
the number $\iota_e(S)$ introduced there is the smallest number of positions of $S$ whose symbols determine
both the codewords of $\pi_S(C_X)$ and those of $\pi_S(C_Z)$.
Thus $\iota_e(S)$ plays, for the pair of codes $\pi_S(C_X)$ and $\pi_S(C_Z)$,
the role that the dimension, that is, the size of an information set, plays for a single code.
The subscript $e$ refers to Euclidean orthogonality, as in ${}^{\perp e}$.
As we shall prove in Lemma \ref{lem:Padm} and Theorem \ref{thm:iotaE} below,
$\iota_e(K\setminus I)$ is the number of qudits of $K\setminus I$ that remain
after removing a largest admissible set, and hence it bounds $|J|-|I|$.

\begin{definition}\label{def:iotaE}
For $S\subseteq\{1,\ldots,n\}$, let
\begin{equation*}
  \iota_e(S)=\min\bigl\{|A| : A\subseteq S,\ \dim\pi_A(C_X)=k_X(S),\ \dim\pi_A(C_Z)=k_Z(S)\bigr\}.
\end{equation*}
Equivalently, $\iota_e(S)$ is the minimum of $|A_X\cup A_Z|$ over information sets
$A_X$ of $\pi_S(C_X)$ and $A_Z$ of $\pi_S(C_Z)$.
\end{definition}

\begin{lemma}\label{lem:Padm}
For $P\subseteq R$, we have $\dim\pi_P(W)=2|P|$ if and only if $R\setminus P$
contains an information set of $\pi_R(C_X)$ and one of $\pi_R(C_Z)$.
Hence the maximum (\ref{eq:maxp}) equals $|R|-\iota_e(R)$, and
the sets $P$ attaining it are the maximum common independent sets of the
column matroids of generator matrices of $W_Z$ and $W_X$.
\end{lemma}
\begin{proof}
Since $\pi_P(W)=\pi_P(W_Z)\times\pi_P(W_X)$ and both factors have dimension at most $|P|$,
$\dim\pi_P(W)=2|P|$ if and only if $\dim\pi_P(W_Z)=\dim\pi_P(W_X)=|P|$.
By the Euclidean analogue of (\ref{eq:dual}),
inside $\mathbb{F}_q^{|R|}$ we have $W_Z=\pi_R(C_Z)^{\perp e}$, and inside $\mathbb{F}_q^{|P|}$
we have $\pi_P(E^{\perp e})^{\perp e}=\sigma_P(E)$ for every linear $E\subseteq \mathbb{F}_q^{|R|}$.
Hence $\dim\pi_P(W_Z)=|P|-\dim\sigma_P(\pi_R(C_Z))$.
It equals $|P|$ if and only if no nonzero word of $\pi_R(C_Z)$ vanishes on $R\setminus P$,
that is, $\pi_{R\setminus P}$ is injective on $\pi_R(C_Z)$,
that is, $R\setminus P$ contains an information set of $\pi_R(C_Z)$.
The same argument applies to $W_X$ and $C_X$.
The last claim is a restatement of the first one.
\end{proof}

The next theorem is the main result of this subsection.
It shows that $\iota_e(K\setminus I)$ bounds $|J|-|I|$,
gives a min--max formula for $\iota_e$ derived from the matroid intersection theorem,
and compares $\iota_e$ with simpler closed-form bounds.

\begin{theorem}\label{thm:iotaE}
Let $C=C_X\times C_Z$ and  $\emptyset\neq I\subsetneq K$ with $K$ a repair group for $I$, and recall that $R=K\setminus I$. Consider a set $P$ attaining the maximum (\ref{eq:maxp}), and let $J$ be computed by the procedure in
Section \ref{sec3} with this $P$ (and any complement $V$). Then:
\begin{enumerate}
\item\label{t1} $|J|-|I|\le\iota_e(K\setminus I)$, and $I\cup(R\setminus P)$ is also a repair group for $I$.
\item\label{t2} We have
\begin{align*}
  \iota_e(R)&=\max_{T\subseteq R}\Bigl\{\dim\sigma_T(\pi_R(C_X))+\dim\sigma_{R\setminus T}(\pi_R(C_Z))\Bigr\}\\
  &=k_X(R)+k_Z(R)-\mu(R),
\end{align*}
where
\begin{equation*}
  \mu(R)=\max\bigl\{|Y| : Y\subseteq R,\ \dim\pi_Y(C_X)=\dim\pi_Y(C_Z)=|Y|\bigr\}.
\end{equation*}
\item\label{t3} We have
\begin{equation*}
  \max\bigl\{k_X(R),k_Z(R)\bigr\}\le\iota_e(R)\le\min\bigl\{|R|,\ \dim\pi_R(C_X+C_Z)\bigr\}.
\end{equation*}
Moreover, by rank--nullity,
\begin{equation}\label{eq:boundJE}
  |J|-|I|\le\dim\pi_R(C_X+C_Z)=\dim\pi_K(C_X+C_Z)-\dim\sigma_I(\pi_K(C_X+C_Z)).
\end{equation}
\item\label{t4} If $\pi_R(C_X)\subseteq\pi_R(C_Z)$ or $\pi_R(C_Z)\subseteq\pi_R(C_X)$, then
$$\iota_e(R)=\max\{k_X(R),k_Z(R)\}=\dim\pi_R(C_X+C_Z).$$
\item\label{t5} The minimality of observables in Item (\ref{l3}) of Theorem \ref{thm1} is retained,
independently of the choice of $P$.
\end{enumerate}
\end{theorem}
\begin{proof}
We begin by showing Item (\ref{t1}). By Item (\ref{l4}) of Theorem \ref{thm1} and Lemma \ref{lem:Padm},
$$|J|\le|K|-|P|=|I|+\iota_e(R).$$
The rows of the matrices $B_I$ and $\widetilde{B}_V$, introduced in Step (1) and in (\ref{eq:elimination_step}) of Section \ref{sec3}, respectively, lie in $\sigma_J(C^{\perp s})$, and
$\pi_I(\sigma_K(C^{\perp s}))=\pi_I(\langle B_I\rangle+\widetilde{V})$ because $\pi_I(W)=\{\mathbf{0}\}$
and $\pi_I(\widetilde{V})=\pi_I(V)$.
Hence $$\pi_I(\sigma_J(C^{\perp s}))\supseteq\pi_I(\sigma_K(C^{\perp s}))=\pi_I(C+C^{\perp s}),$$
by (\ref{eq1b}) for $K$. The reverse inclusion always holds.
Thus $J$ satisfies (\ref{eq1b}) and is a repair group for $I$.
If $J\subseteq J'\subseteq K$, then
$\pi_I(\sigma_J(C^{\perp s}))\subseteq\pi_I(\sigma_{J'}(C^{\perp s}))\subseteq\pi_I(C+C^{\perp s})$.
Therefore, $J'$ satisfies (\ref{eq1b}) as well. The second claim follows by applying this to $J'=K\setminus P=I\cup(R\setminus P)$.

Now, we prove Item (\ref{t2}). By Proposition \ref{prop:edmonds},
the maximum in (\ref{eq:maxp}) is
\begin{equation*}
  \min_{T\subseteq R}\Bigl\{\dim\pi_{R\setminus T}(W_Z)+\dim\pi_T(W_X)\Bigr\}.
\end{equation*}
By the proof of Lemma \ref{lem:Padm}, $\dim\pi_S(W_Z)=|S|-\dim\sigma_S(\pi_R(C_Z))$, and similarly for $X$.
Substituting and using Lemma \ref{lem:Padm} gives the first formula.

For the second one, information sets of $\pi_R(C_X)$ are the bases of the column matroid of a
generator matrix of $C_X$ restricted to $R$, and similarly for $Z$.
Hence $\min|A_X\cup A_Z|=k_X(R)+k_Z(R)-\max|A_X\cap A_Z|$,
and $\max|A_X\cap A_Z|=\mu(R)$ because every common independent set extends to a basis of each matroid.

With respect to Item (\ref{t3}), note that a set $A$ in Definition \ref{def:iotaE} contains an information set of each code, which gives the
lower bound.
Clearly $\iota_e(R)\le|R|$.
Let $A$ be an information set of $$\pi_R(C_X+C_Z)=\pi_R(C_X)+\pi_R(C_Z).$$
Then $\pi_A$ is injective on both summands, and hence $\iota_e(R)\le|A|=\dim\pi_R(C_X+C_Z)$.
By rank--nullity, this is the right-hand side of (\ref{eq:boundJE}), and Item (\ref{t1}) gives the inequality.

Next we show Item (\ref{t4}). If $\pi_R(C_X)\subseteq\pi_R(C_Z)$, an information set of $\pi_R(C_Z)$ is injective on $\pi_R(C_X)$.
Hence $\iota_e(R)\le k_Z(R)=\dim\pi_R(C_X+C_Z)$, and Item (\ref{t3}) gives equality.

Finally, the proof of Item (\ref{l3}) of Theorem \ref{thm1} does not depend on $P$.
\end{proof}

A simpler choice of $P$ only uses the common part $W_Z\cap W_X$ of the two factors of $W$.
The next remark shows that this choice already gives the bound $\dim\pi_{K\setminus I}(C_X+C_Z)$
of Item (\ref{t3}) of Theorem \ref{thm:iotaE},
and that this bound recovers known results of \citep{ilocality26} and \citep{EAQLR-r}.

\begin{remark}\label{rem:comparison_ref30}
Let $W_{XZ}=\sigma_R(C_X^{\perp e}\cap C_Z^{\perp e})$.
Its pivot columns form an admissible set of size
$\dim W_{XZ}=|R|-\dim\pi_R(C_X+C_Z)$, because $(C_X^{\perp e}\cap C_Z^{\perp e})^{\perp e}=C_X+C_Z$.
Hence the bound (\ref{eq:boundJE}) is obtained already with $P$ inside $W_{XZ}$,
and Theorem \ref{thm:iotaE} shows that a maximum $P$ never gives a worse bound.
For $C_X=C_Z$ and a repair group $K$, (\ref{eq1}) gives $\dim\sigma_I(\pi_K(C_X))=\dim\sigma_I(C_X^{\perp e}\cap C_X)$,
and then (\ref{eq:boundJE}) reads $|J|-|I|\le\dim\pi_K(C_X)-\dim\sigma_I(C_X^{\perp e}\cap C_X)$,
which coincides with \citep{ilocality26} in the entanglement-unassisted case $C_X^{\perp e}\subsetneq C_X$.
For $|I|=1$ and $K=\{1,\ldots,n\}$,
(\ref{eq:boundJE}) gives $|J|-1\le\min\{n-1,\dim C_X+\dim C_Z-\dim(C_X\cap C_Z)\}$,
which recovers \citep[Theorem 3]{EAQLR-r}.
\end{remark}

By Items (\ref{t3}) and (\ref{t4}) of Theorem \ref{thm:iotaE}, the choice of $P$ inside $W_{XZ}$ (Remark
\ref{rem:comparison_ref30}) loses nothing for nested pairs.
For non-nested pairs the gap $\dim\pi_R(C_X+C_Z)-\iota_e(R)$ can be large.
If $C_X$ and $C_Z$ are maximum distance separable (MDS) codes of dimensions $\kappa_X$ and $\kappa_Z$,
then $\pi_R(C_j)$ is MDS of dimension $\min\{\kappa_j,|R|\}$ for $j\in\{X,Z\}$.
Hence any $\min\{\max\{\kappa_X,\kappa_Z\},|R|\}$ coordinates of $R$ contain an
information set of both codes, and Item (\ref{t3}) gives
\begin{equation*}
  \iota_e(R)=\min\bigl\{\max\{\kappa_X,\kappa_Z\},|R|\bigr\}.
\end{equation*}
For example, let $\mathbb{F}_{13}^*=\{\gamma_1,\ldots,\gamma_{12}\}$,
$C_X=\bigl\{(f(\gamma_j))_j:f\in\mathbb{F}_{13}[x],\ \deg f<4\bigr\}$ and $C_Z=\bigl\{(x^6g(x)|_{x=\gamma_j})_j:g\in\mathbb{F}_{13}[x],\ \deg g<4\bigr\}$,
where $\mathbb{F}_{13}[x]$ denotes the ring of polynomials in one variable $x$ over $\mathbb{F}_{13}$.
For $K=\{1,\ldots,12\}$ and $I=\{1,2\}$, $K$ is a repair group for $I$,
the bound (\ref{eq:boundJE}) gives $|J|-|I|\le 8$, and Theorem \ref{thm:iotaE} gives
$|J|-|I|\le\iota_e(R)=4$.
An exhaustive search shows that $4$ is the smallest possible value of $|J|-|I|$.

The next two examples illustrate Theorem \ref{thm:iotaE}.
In Example \ref{ex:4222} a largest admissible set gives a smallest repair group,
whereas the choice of Remark \ref{rem:comparison_ref30} may return the whole repair group $K$.
Example \ref{ex:table} compares the bounds of Theorem \ref{thm:iotaE}
with the true minimum of $|J|-|I|$ for several codes.

\begin{example}\label{ex:4222}
Let $q=2$, $n=4$ and
\begin{equation*}
  C_X=\langle(1,0,1,1),(0,1,1,1)\rangle,\qquad C_Z=\langle(1,0,1,0),(0,1,0,1)\rangle.
\end{equation*}
Then $C_X\cap C_Z=\{\mathbf{0}\}$, $C_Z^{\perp e}=C_Z$,
$C_X^{\perp e}=\langle(1,1,1,0),(1,1,0,1)\rangle$, and
$Q(C_X\times C_Z)$ is a $[[4,2,2;2]]_2$ code.
Take $K=\{1,2,3,4\}$ and $I=\{1\}$. Then $R=\{2,3,4\}$ and $K$ is a repair group for $I$.
On $R$ we get $W_Z=\langle(1,0,1)\rangle$ and $W_X=\langle(0,1,1)\rangle$.
Hence $W_{XZ}=\{\mathbf{0}\}$ and the choice of $P$ inside $W_{XZ}$ (Remark \ref{rem:comparison_ref30}) eliminates no qudit.
Its guarantee is only $|J|-|I|\le 3=\dim\pi_R(C_X+C_Z)$,
and for some complements $V$ it returns $J=K$.
With the full $W$, the set $P=\{4\}$ is admissible and maximum,
the procedure returns $J=\{1,2,3\}$ for every $V$, and $\iota_e(R)=2$ with $A=\{2,3\}$.
This is optimal: an exhaustive search shows that no repair group for $I$ has fewer than $|I|+2$ qudits.
The same holds for every $I$ repaired by $K$:
the true minimum of $|J|-|I|$ and $\iota_e(K\setminus I)$ are both $2$.
Here $\dim\pi_K(C_X+C_Z)=4>|K|-1$.
\end{example}

\begin{example}\label{ex:table}
Table \ref{tab:compare} considers the MDS example given after Remark \ref{rem:comparison_ref30} (Row 1),
Example \ref{ex:4222} (Row 2) and two further pairs of binary codes (Rows 3 and 4).
It compares $\dim\pi_R(C_X+C_Z)$, $\iota_e(R)$ and
the true \emph{minimum} of $|J|-|I|$ over all repair groups $J$ for $I$ with $I\subseteq J\subseteq K$.
In all rows $K=\{1,\ldots,n\}$ and $R=K\setminus I$.
By Theorem \ref{thm:iotaE}, the minimum is at most $\iota_e(R)$, which is at most $\dim\pi_R(C_X+C_Z)$.
In Row 2, $\iota_e(R)<\dim\pi_R(C_X+C_Z)$ and $\iota_e(R)$ equals the minimum.
In Row 3 the minimum is smaller than $\iota_e(R)=\dim\pi_R(C_X+C_Z)$.
Row 4 shows both strict inequalities.
We computed the values by exhaustive search over all subsets of $R$.
\begin{table}[ht]
\vspace{2ex}
\centering\footnotesize\setlength{\tabcolsep}{3pt}
\begin{tabular}{@{}clllcccc@{}}
\hline
Row & $q$ & $C_X$ & $C_Z$ & $I$ & $\dim\pi_R(C_X+C_Z)$ & $\iota_e(R)$ & minimum\\
\hline
1 & 13 & \multicolumn{2}{l}{MDS example above} & $\{1,2\}$ & $8$ & $4$ & $4$\\
2 & 2 & \multicolumn{2}{l}{Example \ref{ex:4222}} & $\{1\}$ & $3$ & $2$ & $2$\\
3 & 2 & $\langle(1,0,0,1)\rangle$ & $\langle(0,0,1,0),(0,0,0,1)\rangle$ & $\{1\}$ & $2$ & $2$ & $1$\\
4 & 2 & $\langle(1,0,0,0)\rangle$ & $\langle(1,1,0,1),(0,0,1,1)\rangle$ & $\{2\}$ & $3$ & $2$ & $1$\\
\hline
\end{tabular}
\vspace{1.5ex}
\caption{Comparison of the bounds with the true minimum.
Rows 1 and 2 are the MDS example and Example \ref{ex:4222}.}\label{tab:compare}
\end{table}
\end{example}

The next proposition gives the cost of computing $J$ for CSS codes,
for which the matroid parity problem of Remark \ref{rem:AI1} becomes a matroid intersection problem.

\begin{proposition}\label{prop:complexity}
Given a basis of $\sigma_K(C^{\perp s})$, a set $P$ attaining the maximum in (\ref{eq:maxp}) for
$C=C_X\times C_Z$ can be computed with $O(|K|^3\log|K|)$ field operations using ordinary arithmetic,
and with $O(|K|^{2.62})$ field operations
if fast matrix multiplication is used.
Consequently, the set $J$ of Theorem \ref{thm:iotaE} can be computed
with $O(|K|^3\log|K|)$ field operations, or $O(|K|^{2.62})$
with fast matrix multiplication.
\end{proposition}
\begin{proof}
Splitting a basis of $W$ into bases of $W_Z$ and $W_X$ costs $O(|K|^3)$.
By Lemma \ref{lem:Padm}, (\ref{eq:maxp}) is the linear matroid intersection problem
for generator matrices of $W_Z$ and $W_X$, which have at most $|K|$ rows and $|R|\le|K|$ columns.
Therefore, it is the intersection of two linear matroids on the same ground set $R$ of $m\le|K|$ elements,
each of rank $r\le|K|$.
For such a problem, \citet{cunningham86} gives $O(mr^2\log r)$ field operations
with ordinary matrix arithmetic,
and \citet{gabow96} give $O(mr^{1.62})$ with fast matrix multiplication.
This gives $O(|K|^3\log|K|)$ and $O(|K|^{2.62})$.
The remaining steps of Section \ref{sec3} cost $O(|K|^3)$ by Theorem \ref{thm1},
and $O(|K|^{\omega})$ with $\omega<2.38$ with fast matrix multiplication,
since the echelon form, rank and inverse of an $N\times N$ matrix
can be computed in $O(N^\omega)$ operations over any field
(inversion and triangular factorization \citep{bunch74},
rank and echelon form \citep{ibarra82}).
\end{proof}

We now turn the bounds of Theorem \ref{thm:iotaE} into locality parameters of the code.
As explained at the beginning of this section, the erased positions are not known in advance,
and therefore we fix a collection $\mathcal{K}$ of repair groups
(for instance, a repair group containing each qudit, as in quantum locally recoverable codes \citep{qlrc24}),
and we look for a quantity that bounds $|J|-|I|$
for every $K\in\mathcal{K}$ and every set $I$ of erasures repaired by $K$.
By Item (\ref{t3}) of Theorem \ref{thm:iotaE}, $|J|-|I|\le\dim\pi_{K\setminus I}(C_X+C_Z)$
for the set $J$ of that theorem and every nonempty $I\subsetneq K$ repaired by $K$,
and Remark \ref{rem:minK} below shows that the maximum of the right-hand side over all nonempty $I\subseteq K$
(whether or not $K$ is a repair group for $I$)
is $\min\{\dim\pi_K(C_X+C_Z),|K|-1\}$.
This motivates the following quantum counterpart of the classical dimension locality.

\begin{definition}\label{def:ilE}
  For an entanglement-assisted or unassisted
  quantum error-correcting code $Q(C_X \times C_Z)$ with a collection $\mathcal{K}$
  of repair groups, its \emph{dimension locality} is defined as
  \begin{equation*}
    \mathrm{dimloc}_{\mathcal{K}}\,Q(C_X\times C_Z):=\max_{K \in \mathcal{K}} \min\bigl\{\dim \pi_K(C_X+C_Z),\ |K|-1\bigr\}.
  \end{equation*}
\end{definition}

The next remark justifies Definition \ref{def:ilE}:
it proves the claim on the maximum used above,
and it explains why, unlike in the classical case, the minimum with $|K|-1$ is needed for the sum code $C_X+C_Z$.

\begin{remark}\label{rem:minK}
By Item (\ref{t3}) of Theorem \ref{thm:iotaE} and (\ref{eq:boundJE}), for every $\emptyset\neq I\subsetneq K$ such that $K$ is a repair group for $I$, the upper bound on $|J|-|I|$ given in that item is $\dim\pi_{K\setminus I}(C_X+C_Z)$.
This expression is defined for every nonempty $I\subseteq K$, also when $K$ is not a repair group for $I$,
and we bound it by its maximum over all nonempty $I\subseteq K$.
For a linear code $E\subseteq\mathbb{F}_q^n$ (and likewise over any field),
\begin{equation}\label{eq:maxdrop}
  \max_{\emptyset\neq I\subseteq K}\dim\pi_{K\setminus I}(E)
  =\min\bigl\{\dim\pi_K(E),\ |K|-1\bigr\}.
\end{equation}
Indeed, $S\mapsto\dim\pi_S(E)$ is nondecreasing, and therefore the maximum is attained at some $I=\{i\}$.
Deleting $i$ lowers $\dim\pi_K(E)$ if and only if the unit vector $\mathbf{e}_i$ lies in $\pi_K(E)$.
If $\pi_K(E)\neq\mathbb{F}_q^{|K|}$, some $\mathbf{e}_i\notin\pi_K(E)$, and the maximum is $\dim\pi_K(E)\le|K|-1$.
If $\pi_K(E)=\mathbb{F}_q^{|K|}$, every deletion gives $|K|-1$.
Hence Definition \ref{def:ilE} bounds $|J|-|I|$ for every set $I$ of erasures repaired by $K$
(where $J$ is the set of Theorem \ref{thm:iotaE} if $I\subsetneq K$, and $J=K$ if $I=K$),
and it never exceeds the symbol locality $|K|-1$.
For the sum code $C_X+C_Z$, unlike in the classical case, the minimum with $|K|-1$ is not automatic.
For a classical code $E$ and a repair group $K$ for a coordinate $i$,
we have $\mathbf{e}_i\notin\pi_K(E)$, and hence $\pi_K(E)\neq\mathbb{F}_q^{|K|}$.
In the quantum case, $\pi_K(C_X+C_Z)=\mathbb{F}_q^{|K|}$ can happen (see Example \ref{ex:4222}).
The value $|K|-1$ is the symbol locality for $\delta=2$.
For a general $\delta$, the corresponding value is $|K|-\delta+1$,
and the definition recalled below only requires $|K|-\delta+1\le r$ (see Remark \ref{rem:rdelta}).
\end{remark}

Recall that a quantum code is a quantum $(r,\delta)$-locally recoverable code \citep[Definition 10]{qlrc24}
if every qudit $i$ belongs to a set $K\subseteq\{1,\ldots,n\}$ with $|K|\le r+\delta-1$
that is a repair group for every set $I\subseteq K$ of $\delta-1$ erasures
(see \citep{eaqlrc26} for the entanglement-assisted case).
Implicitly, $|K|\ge\delta-1$ (and $|K|\ge\delta$ if, as in \citep[Definition 9]{qlrc24}, $I\subsetneq K$ is required),
since otherwise the condition on $K$ is vacuous.
In particular, such a set $K$ allows us to recover any $\delta-1$ erasures in $K$
from the remaining $|K|-(\delta-1)\le r$ qudits of $K$.
For entanglement-unassisted CSS codes, the recent preprint \citep{li2026minloc} characterizes such sets $K$
and the least $r$ for which the code has $(r,\delta)$-locality,
whereas our bounds refer to a prescribed collection $\mathcal{K}$ and to the actual set of erasures,
and they also apply to EAQECCs.
The next remark explains why, for $(r,\delta)$-locality with $\delta>2$,
we do not replace $|K|-1$ by $|K|-(\delta-1)$ in Definition \ref{def:ilE}.

\begin{remark}\label{rem:rdelta}
This remark has two parts.
In Part (a) we prove an equality for linear codes,
and in Part (b) we use it to explain the choice of $|K|-1$ in Definition \ref{def:ilE}.

\emph{(a)}
Let $E\subseteq\mathbb{F}_q^{|K|}$ be a linear code whose coordinates are indexed by $K$, and let $1\le t\le|K|-1$.
Then
\begin{equation*}
\max_{I\subseteq K,\ |I|=t}\dim\pi_{K\setminus I}(E)=\min\{\dim\pi_K(E),|K|-t\}.
\end{equation*}
We first prove the inequality $\le$.
For every $I\subseteq K$ with $|I|=t$, puncturing does not increase the dimension,
and the dimension of a code of length $|K\setminus I|$ is at most $|K\setminus I|$.
Hence
\begin{equation*}
\dim\pi_{K\setminus I}(E)\le\min\{\dim\pi_K(E),|K\setminus I|\}=\min\{\dim\pi_K(E),|K|-t\}.
\end{equation*}
For the reverse inequality, we exhibit a set $I$ with $|I|=t$ attaining the right-hand side.
If $\dim\pi_K(E)\le|K|-t$, take $S\subseteq K$ with $|S|=|K|-t$ that contains an information set of $\pi_K(E)$.
Then $\dim\pi_S(E)=\dim\pi_K(E)$.
Otherwise, take $S\subseteq K$ with $|S|=|K|-t$ that is contained in an information set of $\pi_K(E)$.
Then the columns indexed by $S$ of a generator matrix of $\pi_K(E)$ are linearly independent,
and $\dim\pi_S(E)=|S|=|K|-t$.
In both cases, $I=K\setminus S$ attains the maximum.

\emph{(b)}
By Item (\ref{t3}) of Theorem \ref{thm:iotaE}, for every set $I$ of erasures repaired by $K$, with $J$ as in Remark \ref{rem:minK},
\begin{equation*}
|J|-|I|\le\dim\pi_{K\setminus I}(C_X+C_Z).
\end{equation*}
Assume that $K$ is a repair group for every set of $t=\delta-1$ erasures contained in $K$.
By Part (a), applied to $E=\pi_K(C_X+C_Z)$, the maximum of the right-hand side over the sets $I$ with $|I|=\delta-1$
is $\min\{\dim\pi_K(C_X+C_Z),|K|-(\delta-1)\}$.
However, for the sets $I$ with $1\le|I|<\delta-1$, Part (a), applied with $t=|I|$, shows that the maximum of the same right-hand side
is $\min\{\dim\pi_K(C_X+C_Z),|K|-|I|\}$, which can be larger.
Therefore, a definition with $|K|-(\delta-1)$ is guaranteed to bound the right-hand side only for $|I|\ge\delta-1$
(for $|I|<\delta-1$ this happens if and only if $\dim\pi_K(C_X+C_Z)\le|K|-(\delta-1)$),
whereas Definition \ref{def:ilE} bounds it for every nonempty $I\subseteq K$.
\end{remark}

Definition \ref{def:ilE} is based on the bound of Item (\ref{t3}) of Theorem \ref{thm:iotaE}.
Using instead the sharper bound $\iota_e(K\setminus I)$ of Item (\ref{t1}),
and taking the worst case over $I$, we obtain the following refined notion.

\begin{definition}\label{def:ilE2}
For a repair group $K$, define
\begin{equation*}
  \lambda_e(K)=\max_{S\subsetneq K}\iota_e(S)=\max_{\emptyset\neq I\subseteq K}\iota_e(K\setminus I).
\end{equation*}
For $Q(C_X\times C_Z)$ with a collection $\mathcal{K}$ of repair groups,
its {\it refined dimension locality} is

\begin{equation*}
\mathrm{refdimloc}_{\mathcal{K}}\,Q(C_X\times C_Z):=\max_{K\in\mathcal{K}}\lambda_e(K).
\end{equation*}
\end{definition}

\begin{example}\label{ex:4222loc}
We compare the above notions of locality for the code of Example \ref{ex:4222},
where $K=\{1,2,3,4\}$ and $\dim\pi_K(C_X+C_Z)=4>|K|-1$.
For $\mathcal{K}=\{K\}$, the expression of Definition \ref{def:ilE} without the minimum gives $4$,
the symbol locality and Definition \ref{def:ilE} give $3$,
and Definition \ref{def:ilE2} gives $2$.
Every qudit also has a repair group of size $3$, and with those repair groups the symbol locality,
Definition \ref{def:ilE} and Definition \ref{def:ilE2} all give $2$.
\end{example}

The next proposition shows that the refined dimension locality also bounds $|J|-|I|$,
and compares it with Definition \ref{def:ilE} and with the symbol locality.

\begin{proposition}\label{prop:ilE2} The following statements hold:
\begin{enumerate}
\item\label{d1} For every $K\in\mathcal{K}$ and every nonempty $I\subseteq K$ such that $K$ is a repair group for $I$,
the set $J$ of Theorem \ref{thm:iotaE} satisfies $$|J|-|I|\le\lambda_e(K)$$
(for $I=K$ take $J=K$).
\item\label{d2} We have $$\min\{\max\{k_X(K),k_Z(K)\},|K|-1\}\le\lambda_e(K)\le\min\{\dim\pi_K(C_X+C_Z),\ |K|-1\}.$$
\item\label{d3} If $C_X\subseteq C_Z$ or $C_Z\subseteq C_X$, then
$\lambda_e(K)=\min\{\dim\pi_K(C_X+C_Z),|K|-1\}$, and therefore the refined dimension locality equals the dimension locality of Definition \ref{def:ilE}.
\end{enumerate}
\end{proposition}
\begin{proof}
Item (\ref{d1}) follows from Item (\ref{t1}) of Theorem \ref{thm:iotaE}.

For Item (\ref{d2}), Item (\ref{t3}) of Theorem \ref{thm:iotaE} and (\ref{eq:maxdrop}) give
$$\lambda_e(K)\le\max_{S\subsetneq K}\dim\pi_S(C_X+C_Z)=\min\{\dim\pi_K(C_X+C_Z),|K|-1\},$$
and also
$\lambda_e(K)\ge\max_{S\subsetneq K}k_X(S)=\min\{k_X(K),|K|-1\}$. The same holds for $Z$.

For Item (\ref{d3}), nesting gives $\pi_S(C_X)\subseteq\pi_S(C_Z)$ (or the reverse) for every $S$,
which implies $\iota_e(S)=\dim\pi_S(C_X+C_Z)$ by Item (\ref{t4}) of Theorem \ref{thm:iotaE}.
Now take the maximum over $S\subsetneq K$ and use (\ref{eq:maxdrop}).
\end{proof}

\begin{remark}
The maximum in Definition \ref{def:ilE2} cannot be restricted to $|S|=|K|-1$,
because $\iota_e$ is not monotone.
For $C_X=\langle(1,1,1,0)\rangle$ and $C_Z=\langle(1,1,0,1)\rangle$ over $\mathbb{F}_3$
(where $C\cap C^{\perp s}=\{\mathbf{0}\}$, $c=3$ and the code encodes one qudit) we have
$\iota_e(S)=1$ for every $S$ with $|S|=3$, but $\iota_e(\{3,4\})=2$.
Computing $\lambda_e(K)$ requires at most $2^{|K|}-1$ matroid intersections.
However, it is a parameter of the code and is computed only once,
whereas a decoder only needs $\iota_e(K\setminus I)$ for the actual set $I$.
\end{remark}

\subsection{Hermitian orthogonality}\label{sec42}
We now consider EAQECCs from Hermitian orthogonality.
The difference from Subsection \ref{sec41} is that the two symplectic coordinates of each qudit
are combined into a single symbol of $\mathbb{F}_{q^2}$.
As we show below, only one matroid is then involved.
Therefore, a largest admissible set is found by Gaussian elimination, as for classical codes,
instead of by matroid intersection,
and no refinement analogous to Definition \ref{def:ilE2} is needed.

Let $D\subseteq\mathbb{F}_{q^2}^n$ be an $\mathbb{F}_{q^2}$-linear space,
and let $C\subseteq\mathbb{F}_q^{2n}$ correspond to $D$
under the $\mathbb{F}_q$-linear identification of $\mathbb{F}_q^2$ with $\mathbb{F}_{q^2}$
given by a normal basis, as in \citep{ilocality26}.
Then $C^{\perp s}$ corresponds to the Hermitian dual $D^{\perp h}$,
and we assume $C\supsetneq C\cap C^{\perp s}$ as before.
The EAQECC $Q(C)$ has parameters
$[[n, c+2\dim_{\mathbb{F}_{q^2}} D  - n, d; c]]_q$,
where $d$ is the minimum Hamming weight of $D\setminus D^{\perp h}$ and
$c = \dim_{\mathbb{F}_{q^2}} D^{\perp h} - \dim_{\mathbb{F}_{q^2}} (D^{\perp h} \cap D)$.

The map $\beta$ of Section \ref{sec2} turns the observables of \citep{ilocality26}
into those of the EAQECC, and hence only the choice of $P$ needs a comment.
Since $W$ is $\mathbb{F}_{q^2}$-linear,
$\dim_{\mathbb{F}_q}\pi_P(W)=2\dim_{\mathbb{F}_{q^2}}\pi_P(W)$,
and therefore $P$ is admissible if and only if the columns indexed by $P$ of a generator matrix
of $\sigma_R(D^{\perp h})$ over $\mathbb{F}_{q^2}$ are independent, where, as in Subsection \ref{sec41}, $R=K\setminus I$.
Thus we have a single linear matroid.
Let $P_h$ be the set of positions of the pivot columns of the reduced row echelon form
of a generator matrix of $\sigma_R(D^{\perp h})$ over $\mathbb{F}_{q^2}$.
Then $P_h$ is a maximum admissible set, and $|P_h|=\dim_{\mathbb{F}_{q^2}}\sigma_R(D^{\perp h})$.
We also have
\begin{align}
  \dim_{\mathbb{F}_{q^2}} \sigma_{K \setminus I} (D^{\perp h})
  &= (|K| - |I|) - \dim_{\mathbb{F}_{q^2}} \pi_{K\setminus I}(D)\notag\\
  &= (|K|-|I|) - (\dim_{\mathbb{F}_{q^2}} \pi_K(D) - \dim_{\mathbb{F}_{q^2}} \sigma_I(\pi_K(D)))\label{eq201}\\
  &= (|K|-|I|) - (\dim_{\mathbb{F}_{q^2}} \pi_K(D) - \dim_{\mathbb{F}_{q^2}} \sigma_I(D^{\perp h} \cap D)), \label{eq202}
\end{align}
where (\ref{eq201}) holds because the kernel of the restriction of $\pi_{K\setminus I}$ to $\pi_K(D)$ is
$\sigma_I(\pi_K(D))$
and (\ref{eq202}) follows from (\ref{eq1}).

The next proposition is the Hermitian counterpart of Theorem \ref{thm:iotaE}.
By the first equality of the above computation, the bound that it gives on $|J|-|I|$ is
$|K\setminus I|-|P_h|=\dim_{\mathbb{F}_{q^2}}\pi_{K\setminus I}(D)$,
that is, the dimension of a single punctured code.
This is analogous to the classical case recalled at the beginning of this section,
where $|J|-|I|=\dim\pi_{K\setminus I}(E)$, and it differs from the CSS case,
where the bound $\iota_e(K\setminus I)$ depends on the two codes $\pi_{K\setminus I}(C_X)$ and $\pi_{K\setminus I}(C_Z)$.

\begin{proposition}\label{prop:herm}
  Set $P=P_h$. Then the reduction in the size of the repair group satisfies
  $|K|-|J| \geq \dim_{\mathbb{F}_{q^2}}\sigma_{K \setminus I}(D^{\perp h})$,
  and the number $|J|-|I|$ of additional codeword qudits needed for local recovery
  satisfies
  \begin{equation}
    |J|-|I|  \leq \dim_{\mathbb{F}_{q^2}} \pi_K(D)
    - \dim_{\mathbb{F}_{q^2}} \sigma_I(D^{\perp h} \cap D),
    \label{eq:boundJH}
  \end{equation}
  which coincides with \citep{ilocality26}
  for the entanglement-unassisted case $D^{\perp h} \subsetneq D \subsetneq \mathbb{F}_{q^2}^n$.
  The minimality of observables proved in Item (\ref{l3}) of Theorem \ref{thm1} is
  also retained.
\end{proposition}
\begin{proof}
  For the first claim, apply Theorem \ref{thm1} with $P=P_h$. Then we have $|J|\le|K|-|P_h|$, where
  $|P_h|=\dim_{\mathbb{F}_{q^2}}\sigma_{K\setminus I}(D^{\perp h})$.

  Now, $|J|-|I|\le(|K|-|I|)-|P_h|$, which equals the right-hand side of (\ref{eq:boundJH})
  by (\ref{eq201}) and (\ref{eq202}).
  The last claim is Item (\ref{l3}) of Theorem \ref{thm1}.
\end{proof}

As in Subsection \ref{sec41}, this bound leads to a notion of dimension locality,
which is the Hermitian counterpart of Definition \ref{def:ilE}.

\begin{definition}\label{def:ilH}
  For an entanglement-assisted or unassisted
  quantum error-correcting code $Q(C)$ defined by $D$ as above with a collection $\mathcal{K}$
  of repair groups, its \emph{dimension locality} is defined as
  \begin{equation*}
    \mathrm{dimloc}_{\mathcal{K}}\,Q(C):=\max_{K \in \mathcal{K}} \dim_{\mathbb{F}_{q^2}} \pi_K(D).
  \end{equation*}
\end{definition}
For every $I$ repaired by $K$, (\ref{eq201}) and (\ref{eq202}) show that the right-hand side of
(\ref{eq:boundJH}) is $\dim_{\mathbb{F}_{q^2}}\pi_{K\setminus I}(D)$,
and therefore (\ref{eq:maxdrop}) over $\mathbb{F}_{q^2}$ applies.
Hence Definition \ref{def:ilH} bounds $|J|-|I|$ for every $I$ repaired by $K$,
where $J$ is computed with $P=P_h$ if $I\subsetneq K$ and $J=K$ if $I=K$.
Moreover, the branch $\pi_K(D)=\mathbb{F}_{q^2}^{|K|}$ of the proof of (\ref{eq:maxdrop}),
in which the maximum is $|K|-1$, is excluded when $K$ is a repair group for a nonempty set $I$.
Indeed, as in (\ref{eq202}), Condition (\ref{eq1}) gives $\sigma_I(\pi_K(D))=\sigma_I(D^{\perp h}\cap D)$.
Two vectors of $D^{\perp h}\cap D$ are Hermitian orthogonal and, if they are supported on $I$,
their Hermitian inner product equals that of their projections onto $I$.
Hence $\sigma_I(D^{\perp h}\cap D)$ is a Hermitian self-orthogonal subspace of $\mathbb{F}_{q^2}^{|I|}$,
and its dimension is at most $|I|/2<|I|$, because the Hermitian form is nondegenerate.
Thus $\sigma_I(\pi_K(D))\neq\mathbb{F}_{q^2}^{|I|}$, which implies $\pi_K(D)\neq\mathbb{F}_{q^2}^{|K|}$,
that is, $\dim_{\mathbb{F}_{q^2}}\pi_K(D)\le|K|-1$.
Consequently, the maximum in (\ref{eq:maxdrop}) equals $\dim_{\mathbb{F}_{q^2}}\pi_K(D)$,
and, unlike in Definition \ref{def:ilE}, no minimum with $|K|-1$ is needed in Definition \ref{def:ilH}.
Here only one matroid is involved and all its bases have the same size.
Hence the bound $|K\setminus I|-|P|$ of Proposition \ref{prop:herm}
equals $\dim_{\mathbb{F}_{q^2}}\pi_{K\setminus I}(D)$ for every maximum admissible set $P$,
and it cannot be improved by another choice of $P$.
Therefore, a refined definition, as in Definition \ref{def:ilE2}, gives
the maximum of $\dim_{\mathbb{F}_{q^2}}\pi_S(D)$ over $S\subsetneq K$,
which by (\ref{eq:maxdrop}) and the above argument is the value in Definition \ref{def:ilH}.
However, the set $J$ itself may still depend on $P$.
For instance, the example of Remark \ref{rem:AI1}, regarded over $\mathbb{F}_{25}$, is of Hermitian type,
and there $P=\{2\}$ gives $|J|-|I|=2$, whereas $P=\{3\}$ gives $|J|-|I|=1$.

\section*{Declaration of AI use}
The authors used Google Gemini 3.8 Flash to generate an initial version
of the procedure in Section \ref{sec3}, in response to a prompt requesting
the unification of bounds (\ref{eq:boundJE}) and (\ref{eq:boundJH}) that the third author had previously obtained.
The authors improved this initial version through further mathematical analysis.
The authors used ChatGPT 6 Astra-MAX and Gemini 4.0 Argon for literature search on linear matroids, and for searching for mathematical and citation errors.
The authors used Claude (Anthropic) to check
the statements of Proposition \ref{prop:general} and of Subsection \ref{sec41}
by computer on random examples, and to draft their text.
The paper title was generated using Google Gemini 3.8 Flash based on
preliminary drafts of arXiv:2609.23316v1.
All authors take full responsibility for the content of this paper, including the title.


\end{document}